%% file: EJC_elsevier_final_GB26.tex
\documentclass[final,5p,times,twocolumn]{elsarticle}

\journal{European Journal of Control}

\usepackage{amsmath,amssymb,amsfonts,mathtools,amsthm}
\usepackage{array}

\usepackage{algorithm}
\usepackage{algorithmic}
\usepackage{graphicx}
\usepackage{color}
\usepackage[T1]{fontenc}
\usepackage{textcomp}
\usepackage{microtype}
\usepackage{url}
\usepackage[hidelinks]{hyperref}

\biboptions{sort&compress}

\DeclareUnicodeCharacter{2212}{\textminus}

\theoremstyle{plain}
\newtheorem{theorem}{Theorem}
\newtheorem{lemma}{Lemma}
\newtheorem{proposition}{Proposition}
\newtheorem{corollary}{Corollary}

\theoremstyle{definition}
\newtheorem{definition}{Definition}
\newtheorem{assumption}{Assumption}

\theoremstyle{remark}
\newtheorem{remark}{Remark}

\begin{document}

\begin{frontmatter}

\title{Sparse One-Step-Ahead Optimal Control of Time-Varying Affine Opinion Networks: Tracking and Competitive Games}

\author[ufrj]{Gabriel Gentil}
\author[ufrj]{Amit Bhaya\corref{cor1}}
\cortext[cor1]{Corresponding author}
\ead{amit@nacad.ufrj.br}

\affiliation[ufrj]{
  organization={Department of Electrical Engineering, PEE/COPPE, Federal University of Rio de Janeiro},
  addressline={PO Box 68504},
  city={Rio de Janeiro},
  state={RJ},
  postcode={21945-970},
  country={Brazil}
}

\begin{abstract}
This paper studies resource-limited external influence in time-varying opinion networks when a controller must choose both a small set of agents and a scalar intervention at each update. We use the affine free response, which includes DeGroot and Friedkin--Johnsen dynamics, followed by a direct sparse action. Eliminating the scalar action reduces the one-step problem to a cardinality-constrained support objective whose global optimum is obtained from the largest or smallest residual components. Hence exact sparse one-step-ahead optimal control requires one sort and \(O(|\mathcal T|\log|\mathcal T|)\) operations, where \(\mathcal T\) is the support set.

For periodic affine dynamics, the resulting state-dependent low-rank feedback admits cycle-based ultimate bounds and exact tracking of model-invariant targets. The DeGroot and Friedkin--Johnsen specializations expose the role of target invariance and the affine mismatch term. With several competing external players, each sparse best response retains the same sorting structure and the complete binary--continuous stage game is an exact potential game, so a pure-strategy equilibrium exists at every frozen state. Same-state welfare benchmarks separate loss due to scalar competition from the additional loss due to strategic support selection. Numerical examples validate exact support selection and illustrate switching-order effects, DeGroot/Friedkin--Johnsen tracking, resource tradeoffs, and competitive implementations.
\end{abstract}


\end{frontmatter}

% ...

\section{Introduction}
\label{sec:introduction}

Opinion dynamics provide a networked-systems setting in which consensus,
persistent disagreement, stubbornness, and external influence can be
studied within a common dynamical framework.  The DeGroot model
\cite{DeGroot1974} is a canonical linear consensus process, while the
Friedkin--Johnsen (FJ) model
\cite{FriedkinJohnsen1990,FriedkinJohnsen1999} incorporates persistent
attachment to innate opinions.  Systems-theoretic analyses and extensions
include \cite{ParsegovEtAl2017,ProskurnikovTempo2017,
ProskurnikovTempo2018,AndersonYe2019,TianWang2023}; time-varying FJ-type
models are treated, for example, in
\cite{ProskurnikovEtAl2017TimeVarying,DisaroValcher2024}.  Wang et al.
\cite{WangEtAl2023ConcatenatedFJ} show, in a concatenated FJ construction,
that the order of the same coarse-grained stochastic transition matrices
can determine whether consensus or periodic behavior occurs.

A practical limitation in network intervention is that only a small number
of agents may be reachable at each update.  The resulting control problem
has two coupled decisions: \emph{where} to act and \emph{how strongly} to
act.  Node-selection and influence-allocation problems are therefore
naturally combinatorial; related work includes graph-based targeting,
influence games, and FJ opinion optimization under node budgets
\cite{Masuda2015,BiniEtAl2022,BastopcuEtAl2025,SunZhang2023,SunZhang2026}.
One-step-ahead optimal control (OSAOC) \cite{KaszkurewiczBhaya2022} is
well suited to resource-limited online intervention because it uses the
current state rather than a long prediction horizon.

We write the uncontrolled one-step response as
\[
    z(k)=F(k)x(k)+h(k),
\]
and apply a direct external intervention after this response:
\[
    x(k+1)=F(k)x(k)+h(k)+b(k)u(k).
\]
DeGroot corresponds to \(F(k)=W(k)\), \(h(k)=0\); FJ corresponds to
\(F(k)=\Lambda W(k)\), \(h(k)=(I-\Lambda)s\).  The binary vector \(b(k)\)
selects at most a prescribed number of target agents, while the scalar
\(u(k)\) determines the signed intervention magnitude.  This convention
separates susceptibility to interpersonal influence from direct external
actuation.

The key computational observation is that the opinion model enters the
one-step optimization only through the residual \(c(k)=g-z(k)\).  The
scalar action can be eliminated analytically.  For every support size,
the remaining objective is maximized by either the largest or the smallest
residual components.  A full sort and prefix/suffix sums give the global
sparse solution in \(O(|\mathcal T|\log|\mathcal T|)\) time, while
linear-time extreme selection followed by sorting only the selected
\(r\)-element tails gives \(O(|\mathcal T|+r\log r)\).  The
sorting theorem is consequently an affine one-step result, not a
DeGroot-specific property.

The multiplayer extension builds on two complementary strands.  Previous
competitive OSAOC with prescribed influence domains is developed in
\cite{GentilBhaya2024}.  More recently, Liu, Mazalov, and Gao
\cite{LiuMazalovGao2026EJC} compare OSAOC, feedback, and hybrid strategies
in an average-oriented two-player opinion game.  A companion preprint by the
present authors \cite{GentilBhaya2026CompetitiveFJ} fixes the influence
directions and analyzes the resulting continuous OSAOC game in detail,
including potential structure, best-response implementations, closed-loop
stability, equilibrium geometry, and same-state welfare.  The present
paper addresses the complementary sparse problem: the support itself is a
strategic online decision.  This adds a discrete combinatorial layer while
preserving the exact single-player sorting structure.

For clarity, the fixed-direction results of
\cite{GentilBhaya2026CompetitiveFJ} are cited rather than reproved below.
The novelty in this paper is the extension to support-changing best responses and a
binary--continuous exact-potential game, together with a welfare benchmark
that separates scalar-action loss from the additional loss created by
noncooperative support selection.  This distinction also avoids confusing
structural error imposed by the sparse feasible family with competitive
inefficiency.

The main contributions are:
\begin{enumerate}
    \item \textbf{Exact sparse OSAOC for affine opinion dynamics.}
    The scalar action is eliminated in closed form and the global support
    is obtained by extreme-residual sorting.  A full-sort implementation costs
    \(O(|\mathcal T|\log|\mathcal T|)\), while partial extreme
    selection gives \(O(|\mathcal T|+r\log r)\) and hence linear
    complexity for fixed \(r\).

    \item \textbf{Controller-dependent tracking certificates.}
    One-step optimality implies a global contraction factor on the
    target component of the free-response error.  Combining this factor
    with target/non-target block gains gives an a priori, checkable
    small-gain condition for ultimate boundedness and exact invariant-target
    tracking, independent of the realized support sequence.  For periodic
    dynamics we retain complementary trajectory-wise cycle bounds and
    identify precisely what the control-independent norm certificate can
    and cannot establish.

    \item \textbf{Competitive sparse support selection.}
    Each player's support-and-action best response is obtained by the same
    sorting rule.  The complete binary--continuous stage game remains an
    exact potential game and therefore admits a pure-strategy equilibrium
    at every frozen network state.  These are frozen-state stage-game
    statements; no general convergence or stability claim is made for the
    competitive closed loop generated by repeated equilibrium or
    best-response implementation.

    \item \textbf{Sparse welfare decomposition and numerical validation.}
    Same-state benchmarks split competitive loss into scalar-action and
    support-selection components.  DeGroot and FJ experiments illustrate
    model-independent sparse computation versus model-dependent tracking,
    while competitive experiments distinguish equilibrium from one-sweep
    implementations.
\end{enumerate}

The remainder of the paper develops the single-player sparse problem,
periodic tracking, the competitive sparse game and welfare benchmarks, and
finally the numerical experiments and conclusions.

\section{Sparse One-Step-Ahead Optimal Control}
\label{sec:single_player}

Consider the time-varying affine opinion dynamics
\begin{equation}
    x(k+1)=F(k)x(k)+h(k)+b(k)u(k),
    \label{eq:single_dynamics}
\end{equation}
where \(x(k)\in\mathbb R^n\), \(F(k)\in\mathbb R^{n\times n}\),
\(h(k)\in\mathbb R^n\), \(u(k)\in\mathbb R\) is a scalar direct
external action, and \(b(k)\in\{0,1\}^n\) is a binary support vector.
At decision time \(k\), define the uncontrolled or free one-step response
\[
    z(k):=F(k)x(k)+h(k).
\]
The DeGroot specialization is \(F(k)=W(k)\), \(h(k)=0\), with
row-stochastic \(W(k)\) \cite{DeGroot1974}.  The Friedkin--Johnsen
specialization used here is
\[
    F(k)=\Lambda W(k),
    \qquad
    h(k)=(I-\Lambda)s,
\]
where \(W(k)\) is row-stochastic,
\(\Lambda=\operatorname{diag}(\lambda_1,\ldots,\lambda_n)\),
\(0\le\lambda_i\le1\), and \(s\) is the innate-opinion vector
\cite{FriedkinJohnsen1990,FriedkinJohnsen1999,ParsegovEtAl2017}.
The intervention \(b(k)u(k)\) is applied directly and is not filtered
through \(\Lambda\).  Conditional on a chosen support, the same scalar
amplitude is applied to every selected node.  Thus this is a shared-amplitude
rank-one input model, not an \(\ell_0\)-constrained vector-input problem with
independent node-wise control amplitudes.

Let \(\mathcal T\subseteq\{1,\ldots,n\}\) be the prescribed target set,
let \(g\in\mathbb R^n\) be a desired opinion profile, and let
\(1\le r\le |\mathcal T|\) be the maximum number of target nodes that may be
actuated at one time.  The admissible support set is
\begin{equation}
    \mathcal U_r
    :=
    \left\{
        b\in\{0,1\}^n:
        \operatorname{supp}(b)\subseteq\mathcal T,\;
        \mathbf 1^\top b\le r
    \right\}.
    \label{eq:admissible_supports}
\end{equation}
At time \(k\), sparse OSAOC minimizes the one-step regularized target
error
\begin{equation}
    J_k(b,u)
    :=
    \sum_{i\in\mathcal T}
    \left[
        g_i-\bigl(z(k)+bu\bigr)_i
    \right]^2
    +\gamma u^2,
    \qquad \gamma>0,
    \label{eq:osaoc_cost}
\end{equation}
over \(b\in\mathcal U_r\) and \(u\in\mathbb R\).

The cost \eqref{eq:osaoc_cost} penalizes only \(\mathcal T\); it therefore
does not by itself bound complementary state components, which are handled by
the propagation analysis of Section~\ref{sec:tracking}.  The shared scalar is
also essential to the combinatorics.  With independent sparse node inputs
\(v_i\), minimizing
\(\sum_{i\in\mathcal T}(c_i-v_i)^2+\gamma\|v\|_2^2\) gives
\(v_i^\star=c_i/(1+\gamma)\) on each active coordinate, so the optimal support
is simply the \(r\) largest \(|c_i|\).  The nontrivial \(|S|+\gamma\) coupling
below is therefore caused by sparsity together with a single shared amplitude.

\paragraph{Susceptibility-filtered actuation.}
The direct-actuation convention is structurally important.  If an FJ
susceptibility matrix also filters the external input,
\[
    x^+=F x+h+\Lambda b u,
\]
then, for a support \(S\subseteq\mathcal T\), eliminating the scalar action gives
\[
    u^\star(S)
    =
    \frac{\sum_{i\in S}\lambda_i c_i}
         {\sum_{i\in S}\lambda_i^2+\gamma},
    \qquad
    V_\Lambda(S)
    =
    \frac{\left(\sum_{i\in S}\lambda_i c_i\right)^2}
         {\sum_{i\in S}\lambda_i^2+\gamma}.
\]
If \(\lambda_i\equiv\lambda>0\) on \(\mathcal T\), then
\(V_\Lambda(S)=(\sum_{i\in S}c_i)^2/(|S|+\gamma/\lambda^2)\), so
Theorem~\ref{thm:sorting} survives verbatim after the rescaling
\(\gamma\mapsto\gamma/\lambda^2\).  For heterogeneous susceptibilities,
however, the denominator depends on the identities of the selected nodes as
well as their cardinality.  The problem becomes a weighted fractional binary
support problem, and the extreme-residual sorting reduction no longer follows.
This gives a second, model-based illustration of the tractability boundary
identified below.

Define the one-step residual
\begin{equation}
    c(k):=g-z(k)=g-F(k)x(k)-h(k).
    \label{eq:residual}
\end{equation}
Only its components on \(\mathcal T\) enter the optimization.  Suppressing
the time index when no ambiguity arises,
\begin{equation}
    J_k(b,u)
    =
    \sum_{i\in\mathcal T}(c_i-b_i u)^2+\gamma u^2.
    \label{eq:cost_residual_form}
\end{equation}

\subsection{Elimination of the Continuous Action}
\label{subsec:continuous_elimination}

For a fixed \(b\in\mathcal U_r\), the cost is a strictly convex
quadratic function of \(u\):
\begin{equation}
\begin{aligned}
    J_k(b,u)
    &=
    \sum_{i\in\mathcal T}c_i^2
    -2u\,b^\top c
    +u^2\bigl(b^\top b+\gamma\bigr).
\end{aligned}
\label{eq:cost_quadratic_u}
\end{equation}
Since \(b\) is binary and supported on \(\mathcal T\),
\[
    b^\top b=\mathbf 1^\top b=|\operatorname{supp}(b)|.
\]
Therefore the unique minimizer is
\begin{equation}
    u^\star(b)
    =
    \frac{b^\top c}{\mathbf 1^\top b+\gamma}.
    \label{eq:optimal_scalar}
\end{equation}
Substituting \eqref{eq:optimal_scalar} into
\eqref{eq:cost_residual_form} gives
\begin{equation}
    J_k^\star(b)
    =
    \sum_{i\in\mathcal T}c_i^2
    -
    \frac{(b^\top c)^2}{\mathbf 1^\top b+\gamma}.
    \label{eq:reduced_cost}
\end{equation}
Hence minimizing \(J_k^\star(b)\) is equivalent to
\begin{equation}
    \max_{b\in\mathcal U_r}
    \frac{(b^\top c)^2}{\mathbf 1^\top b+\gamma}.
    \label{eq:binary_problem}
\end{equation}

For a support \(S\subseteq\mathcal T\), define
\begin{equation}
    V(S)
    :=
    \frac{\left(\sum_{i\in S}c_i\right)^2}{|S|+\gamma},
    \qquad
    V(\varnothing):=0.
    \label{eq:set_objective}
\end{equation}
Problem \eqref{eq:binary_problem} is therefore
\[
    \max_{S\subseteq\mathcal T,\; |S|\le r}V(S).
\]

\subsection{Exact Support Selection by Residual Sorting}
\label{subsec:sorting}

The reduced problem is combinatorial, but its special dependence on
the \emph{sum} of the selected residuals allows an exact solution after
one sort.

Let \(N:=|\mathcal T|\), and let
\[
    c_{(1)}\ge c_{(2)}\ge\cdots\ge c_{(N)}
\]
denote the residual components on \(\mathcal T\) arranged in
nonincreasing order.  Let \(i_{(j)}\in\mathcal T\) be an index attaining
\(c_{(j)}\), with ties resolved arbitrarily.  For
\(s=1,\ldots,r\), define the extreme supports
\begin{equation}
\begin{aligned}
    S_{s,+}
    &:=\{i_{(1)},\ldots,i_{(s)}\},\\
    S_{s,-}
    &:=\{i_{(N-s+1)},\ldots,i_{(N)}\}.
\end{aligned}
\label{eq:extreme_supports}
\end{equation}

\begin{theorem}[Exact support reduction by sorting]
\label{thm:sorting}
For every fixed cardinality \(s\in\{1,\ldots,r\}\), at least one
maximizer of
\[
    \max_{S\subseteq\mathcal T,\;|S|=s}
    \left(\sum_{i\in S}c_i\right)^2
\]
belongs to the two-element candidate family
\(\{S_{s,+},S_{s,-}\}\).  Consequently,
\begin{equation}
    \max_{S\subseteq\mathcal T,\;|S|\le r}V(S)
    =
    \max_{1\le s\le r}
    \max\left\{
        V(S_{s,+}),V(S_{s,-})
    \right\}.
    \label{eq:global_support_reduction}
\end{equation}
An optimal sparse OSAOC support can therefore be computed by one full sort
in \(O(N\log N)\) time.  Alternatively, selecting the \(r\) largest
and \(r\) smallest residuals in linear time and sorting only those
extremes gives \(O(N+r\log r)\), hence \(O(N)\) for fixed \(r\).
\end{theorem}

\begin{proof}
Fix \(s\).  Since \(s+\gamma\) is constant, maximizing \(V(S)\)
over \(|S|=s\) is equivalent to maximizing
\(\left|\sum_{i\in S}c_i\right|\).
Among all \(s\)-element subsets, the largest possible signed sum is
obtained by selecting the \(s\) largest components,
namely \(S_{s,+}\), while the smallest possible signed sum is obtained
by selecting the \(s\) smallest components, namely \(S_{s,-}\).
Therefore the largest absolute subset sum, and hence the largest
squared subset sum, is attained by at least one of these two supports.

It remains to compare the two candidates for every
\(s=1,\ldots,r\).  Since \(r\ge1\) and \(V(S)\ge0\), the empty support
cannot improve on all nonempty candidates; when every residual is zero,
it merely ties them.  A full sort of the \(N\) residual components costs
\(O(N\log N)\), after which all candidate sums require only \(O(r)\)
prefix/suffix updates.  Alternatively, linear-time selection can identify
the \(r\) largest and \(r\) smallest components; sorting only these two
tails costs \(O(r\log r)\), after which the same cumulative-sum sweep
applies.  This gives \(O(N+r\log r)\), which is linear in \(N\) when
\(r\) is fixed.
\end{proof}

The tractability hinges not only on extreme subset sums but on the denominator:
for fixed \(s\), \(|S|+\gamma=s+\gamma\) is independent of which nodes are
selected.  This decouples the cardinality sweep from the extreme-sum problem.
A support-dependent denominator such as \(b^\top Qb+\gamma\) would destroy
this two-candidate reduction, marking the boundary exploited by
Theorem~\ref{thm:sorting}.

\begin{remark}[Ties and nonuniqueness]
\label{rem:ties}
Theorem~\ref{thm:sorting} guarantees an optimal support among the
extreme candidates in \eqref{eq:extreme_supports}; it does not imply
uniqueness.  Equal residual components may generate several distinct
optimal supports with the same objective value.
\end{remark}

The cumulative-sum implementation is explicit.  Define
\begin{equation}
    p_s:=\sum_{j=1}^{s}c_{(j)},
    \qquad
    q_s:=\sum_{j=N-s+1}^{N}c_{(j)}.
    \label{eq:prefix_suffix_sums}
\end{equation}
Then
\begin{equation}
    V(S_{s,+})=\frac{p_s^2}{s+\gamma},
    \qquad
    V(S_{s,-})=\frac{q_s^2}{s+\gamma},
    \label{eq:prefix_suffix_values}
\end{equation}
with the recursions
\[
    p_s=p_{s-1}+c_{(s)},
    \qquad
    q_s=q_{s-1}+c_{(N-s+1)}.
\]

\begin{algorithm}[t]
\caption{Residual-sort implementation of sparse OSAOC}
\label{alg:sparse_osaoc}
\begin{algorithmic}[1]
\REQUIRE residual \(c(k)\), target set \(\mathcal T\), budget \(r\),
         penalty \(\gamma>0\), numerical dead-band \(\tau_{\rm act}\ge0\)
\ENSURE support \(S^\star\), binary vector \(b^\star(k)\),
        and scalar action \(u^\star(k)\)
\STATE Sort \(\{c_i(k)\}_{i\in\mathcal T}\) so that
       \(c_{(1)}\ge\cdots\ge c_{(N)}\)
\STATE \(p\gets0,\;q\gets0,\;V_{\max}\gets\tau_{\rm act},\;S^\star\gets\varnothing\)
\FOR{\(s=1,\ldots,r\)}
    \STATE \(p\gets p+c_{(s)}\)
    \STATE \(q\gets q+c_{(N-s+1)}\)
    \STATE \(V_+\gets p^2/(s+\gamma)\)
    \STATE \(V_-\gets q^2/(s+\gamma)\)
    \IF{\(V_+>V_{\max}\)}
        \STATE \(V_{\max}\gets V_+\),
               \(S^\star\gets S_{s,+}\)
    \ENDIF
    \IF{\(V_->V_{\max}\)}
        \STATE \(V_{\max}\gets V_-\),
               \(S^\star\gets S_{s,-}\)
    \ENDIF
\ENDFOR
\STATE Set \(b_i^\star(k)=1\) for \(i\in S^\star\) and
       \(b_i^\star(k)=0\) otherwise
\STATE \(u^\star(k)\gets
       \left(\sum_{i\in S^\star}c_i(k)\right)/
       \left(|S^\star|+\gamma\right)\)
\RETURN \(S^\star,b^\star(k),u^\star(k)\)
\end{algorithmic}
\end{algorithm}

With \(\tau_{\rm act}=0\), Algorithm~\ref{alg:sparse_osaoc} is the exact
optimizer characterized by Theorem~\ref{thm:sorting}.  A positive
\(\tau_{\rm act}\) is only a floating-point safeguard: improvements below
that threshold are treated as numerically zero and the empty support is
returned.  Numerical experiments must therefore report the value of
\(\tau_{\rm act}\); it is not used in the exactness claims.  All analytical
results below take \(\tau_{\rm act}=0\).

\subsection{When Is a Larger Support Beneficial?}
\label{subsec:marginal_support}

The cardinality constraint is an upper bound rather than an equality,
so the optimal controller need not use all \(r\) available actuation
slots.  The following condition characterizes exactly when adding one
node improves the reduced objective.

\begin{theorem}[Marginal improvement condition]
\label{thm:marginal}
Let \(S\subseteq\mathcal T\), let \(s:=|S|\), and define
\[
    a_S:=\sum_{i\in S}c_i.
\]
For \(j\in\mathcal T\setminus S\),
\[
    V(S\cup\{j\})>V(S)
\]
if and only if
\begin{equation}
    (a_S+c_j)^2
    >
    \frac{s+1+\gamma}{s+\gamma}\,a_S^2.
    \label{eq:marginal_condition_1}
\end{equation}
Equivalently,
\begin{equation}
    (s+\gamma)c_j^2
    +2(s+\gamma)a_Sc_j
    -a_S^2
    >0.
    \label{eq:marginal_condition_2}
\end{equation}
\end{theorem}

\begin{proof}
By \eqref{eq:set_objective},
\[
    V(S)=\frac{a_S^2}{s+\gamma},
    \qquad
    V(S\cup\{j\})
    =
    \frac{(a_S+c_j)^2}{s+1+\gamma}.
\]
Since both denominators are positive, cross multiplication gives
\eqref{eq:marginal_condition_1}.  Expanding and collecting terms gives
\eqref{eq:marginal_condition_2}.
\end{proof}

\begin{corollary}[Sign-dependent marginal thresholds]
\label{cor:same_sign_threshold}
Let \(a_S\neq0\).  If \(a_Sc_j\ge0\), then adding node \(j\) strictly
improves the reduced objective if and only if
\begin{equation}
    |c_j|
    >
    |a_S|
    \left(
        \sqrt{1+\frac{1}{s+\gamma}}-1
    \right).
    \label{eq:same_sign_threshold}
\end{equation}
If \(a_Sc_j<0\), the corresponding condition is
\[
    |c_j|
    >
    |a_S|
    \left(
        1+\sqrt{1+\frac{1}{s+\gamma}}
    \right).
\]
If \(a_S=0\), every node with \(c_j\neq0\) gives a strict improvement.
\end{corollary}

\begin{proof}
Let \(\kappa:=\sqrt{1+1/(s+\gamma)}>1\).  Condition
\eqref{eq:marginal_condition_1} is equivalent to
\(|a_S+c_j|>\kappa|a_S|\).  If \(a_Sc_j\ge0\), then
\(|a_S+c_j|=|a_S|+|c_j|\), which gives
\eqref{eq:same_sign_threshold}.  If \(a_Sc_j<0\), the inequality cannot
hold when \(|c_j|\le|a_S|\); for \(|c_j|>|a_S|\) it becomes
\(|c_j|-|a_S|>\kappa|a_S|\), yielding the displayed opposite-sign threshold.  The case \(a_S=0\) follows directly
from \eqref{eq:set_objective}.
\end{proof}

\begin{remark}[Implicit sparsity]
\label{rem:implicit_sparsity}
The penalty \(\gamma u^2\) does not directly penalize the number of
selected nodes.  Nevertheless, the denominator
\(|S|+\gamma\) in \eqref{eq:set_objective} creates an endogenous
tradeoff: increasing the support is useful only if the newly added
residual components increase the squared aggregate residual enough to
offset the larger denominator.  Thus an ``at most \(r\)'' constraint
need not be saturated.
\end{remark}

\subsection{Closed-Loop Form and One-Step Improvement}
\label{subsec:single_closed_loop}

Let \(b^\star(k)\) and \(u^\star(k)\) be obtained from
Algorithm~\ref{alg:sparse_osaoc}, and define
\[
    s^\star(k):=\mathbf 1^\top b^\star(k).
\]
The closed-loop update is
\begin{equation}
    x(k+1)
    =
    z(k)+b^\star(k)u^\star(k)
    =F(k)x(k)+h(k)+b^\star(k)u^\star(k).
    \label{eq:sparse_closed_loop_1}
\end{equation}
Using \eqref{eq:optimal_scalar} and \eqref{eq:residual},
\begin{equation}
\begin{aligned}
    x(k+1)
    &=
    \bigl(I-\Phi(k)\bigr)z(k)+\Phi(k)g\\
    &=
    \bigl(I-\Phi(k)\bigr)\bigl(F(k)x(k)+h(k)\bigr)+\Phi(k)g,
\end{aligned}
    \label{eq:sparse_closed_loop_2}
\end{equation}
where
\begin{equation}
    \Phi(k)
    :=
    \frac{
        b^\star(k)b^\star(k)^\top
    }{
        s^\star(k)+\gamma
    }.
    \label{eq:phi_definition}
\end{equation}
For \(s^\star(k)=0\), \(\Phi(k)=0\).

The next proposition states precisely what one-step optimality
guarantees.

\begin{proposition}[One-step improvement relative to zero control]
\label{prop:one_step_improvement}
Let
\[
    E_0(k)
    :=
    \sum_{i\in\mathcal T}c_i(k)^2
\]
be the target prediction error obtained by applying zero control at
time \(k\), and let
\[
    E_\star(k)
    :=
    \sum_{i\in\mathcal T}
    \left(c_i(k)-b_i^\star(k)u^\star(k)\right)^2
\]
be the target prediction error under sparse OSAOC.  Then
\begin{equation}
    E_\star(k)+\gamma\,u^\star(k)^2
    =
    E_0(k)
    -
    \frac{
        \bigl(b^\star(k)^\top c(k)\bigr)^2
    }{
        s^\star(k)+\gamma
    }
    \le E_0(k),
    \label{eq:one_step_cost_reduction}
\end{equation}
and therefore
\begin{equation}
    E_\star(k)\le E_0(k).
    \label{eq:one_step_error_nonworsening}
\end{equation}
The final inequality in \eqref{eq:one_step_cost_reduction} is tight if and only if
\(b^\star(k)^\top c(k)=0\), in which case
\(u^\star(k)=0\).
\end{proposition}

\begin{proof}
Equation \eqref{eq:one_step_cost_reduction} follows directly from the
reduced optimal cost \eqref{eq:reduced_cost}.  Since
\(\gamma u^\star(k)^2\ge0\), inequality
\eqref{eq:one_step_error_nonworsening} follows immediately.
\end{proof}

\begin{remark}[No temporal monotonicity is implied]
\label{rem:no_temporal_monotonicity}
Proposition~\ref{prop:one_step_improvement} compares the controlled and
uncontrolled \emph{predictions from the same state \(x(k)\)}.  It does
not assert that a tracking-error norm is monotonically decreasing from
time \(k\) to time \(k+1\).  Temporal convergence requires a separate
closed-loop analysis, developed in Section~\ref{sec:tracking}.
\end{remark}

The matrix \(\Phi(k)\) has a useful geometric interpretation.

\begin{proposition}[Regularized rank-one correction]
\label{prop:scaled_projector}
If \(s^\star(k)>0\), define
\begin{equation}
    P_b(k)
    :=
    \frac{
        b^\star(k)b^\star(k)^\top
    }{
        s^\star(k)
    }.
    \label{eq:orthogonal_projector_b}
\end{equation}
Then \(P_b(k)\) is the orthogonal projector onto
\(\operatorname{span}\{b^\star(k)\}\), and
\begin{equation}
    \Phi(k)
    =
    \frac{s^\star(k)}{s^\star(k)+\gamma}\,P_b(k).
    \label{eq:scaled_projector_identity}
\end{equation}
Consequently, \(\Phi(k)\) is symmetric positive semidefinite,
has rank one, and has nonzero eigenvalue
\begin{equation}
    \lambda_{\Phi}(k)
    =
    \frac{s^\star(k)}{s^\star(k)+\gamma}
    \in(0,1).
    \label{eq:phi_eigenvalue}
\end{equation}
Thus sparse OSAOC applies a state-dependent \emph{scaled} rank-one
projection correction; \(\Phi(k)\) itself is not an idempotent
projector when \(\gamma>0\).
\end{proposition}

\begin{proof}
Since
\(\|b^\star(k)\|_2^2=s^\star(k)\),
\[
    P_b(k)^2
    =
    \frac{
        b^\star b^{\star\top}b^\star b^{\star\top}
    }{
        (s^\star)^2
    }
    =
    \frac{
        s^\star b^\star b^{\star\top}
    }{
        (s^\star)^2
    }
    =
    P_b(k),
\]
and \(P_b(k)=P_b(k)^\top\).  Hence it is the orthogonal projector onto
\(\operatorname{span}\{b^\star(k)\}\).  Identity
\eqref{eq:scaled_projector_identity} follows from
\eqref{eq:phi_definition}, and the remaining claims are immediate.
\end{proof}

\section{Tracking over Time-Varying Affine Opinion Networks}
\label{sec:tracking}

The sparse OSAOC support depends on the current residual and therefore
on the state.  We first exploit one-step optimality itself to obtain an
a priori controller-dependent comparison bound that does not require
periodicity and does not depend on the realized support sequence.  We then
specialize to periodic free dynamics and derive complementary cycle bounds.
Even in that case the closed-loop matrices need not be periodic, because
the selected support remains state dependent.

Let
\begin{equation}
    e(k):=x(k)-g.
    \label{eq:tracking_error}
\end{equation}
From \eqref{eq:sparse_closed_loop_2},
\begin{equation}
\begin{aligned}
    e(k+1)
    &=
    \bigl(I-\Phi(k)\bigr)F(k)e(k)\\
    &\quad+
    \bigl(I-\Phi(k)\bigr)
    \bigl(F(k)g+h(k)-g\bigr).
\end{aligned}
\label{eq:error_dynamics}
\end{equation}
Define
\begin{equation}
    A(k)
    :=
    \bigl(I-\Phi(k)\bigr)F(k),
    \;
    d(k)
    :=
    \bigl(I-\Phi(k)\bigr)
    \bigl(F(k)g+h(k)-g\bigr),
    \label{eq:A_d_definitions}
\end{equation}
so that
\begin{equation}
    e(k+1)=A(k)e(k)+d(k).
    \label{eq:error_affine}
\end{equation}

\subsection{Uniform Bounds Induced by the Regularized Correction}
\label{subsec:uniform_bounds}

The scaled-projector identity in
Proposition~\ref{prop:scaled_projector} immediately yields a useful
norm bound.

\begin{lemma}[Nonexpansiveness of \(I-\Phi(k)\)]
\label{lem:I_minus_phi}
For every \(k\),
\begin{equation}
    \left\|I-\Phi(k)\right\|_2\le 1.
    \label{eq:I_minus_phi_bound}
\end{equation}
If \(s^\star(k)>0\), the eigenvalues of \(I-\Phi(k)\) are
\(1\) on \(\operatorname{span}\{b^\star(k)\}^{\perp}\) and
\(\gamma/(s^\star(k)+\gamma)\) on
\(\operatorname{span}\{b^\star(k)\}\).
\end{lemma}

\begin{proof}
If \(s^\star(k)=0\), then \(\Phi(k)=0\) and the result is immediate.
Otherwise, by \eqref{eq:scaled_projector_identity},
\[
    \Phi(k)
    =
    \frac{s^\star(k)}{s^\star(k)+\gamma}P_b(k),
\]
where \(P_b(k)\) is an orthogonal projector. Hence the eigenvalues of
\(I-\Phi(k)\) are \(1\) and
\(1-s^\star(k)/(s^\star(k)+\gamma)
=\gamma/(s^\star(k)+\gamma)\), both in \((0,1]\).
Since \(I-\Phi(k)\) is symmetric, its induced Euclidean norm is its
largest eigenvalue.
\end{proof}

\subsection{An A Priori Controller-Dependent Small-Gain Bound}
\label{subsec:apriori_small_gain}

The nonexpansiveness result above does not use the fact that
\(b^\star(k)\) is an \emph{optimal} sparse support.  Exploiting that
optimality gives a strict contraction factor on the target component of
the one-step free-response error.

Let \(P_{\mathcal T}\) be the diagonal projector onto the nonempty target
set \(\mathcal T\), let
\(P_{\mathcal R}:=I-P_{\mathcal T}\), where
\(\mathcal R=\{1,\ldots,n\}\setminus\mathcal T\), and set
\begin{equation}
    n_{\mathcal T}:=|\mathcal T|,
    \qquad
    \alpha_{\mathcal T}
    :=
    \sqrt{1-\frac{1}{n_{\mathcal T}(1+\gamma)}}
    \in[0,1).
    \label{eq:target_contraction_factor}
\end{equation}
Throughout this subsection assume \(r\ge1\), so every singleton target
support is admissible.  Define the affine target mismatch
\begin{equation}
    \eta(k):=F(k)g+h(k)-g.
    \label{eq:eta_mismatch}
\end{equation}

\begin{theorem}[Controller-dependent comparison inequality]
\label{thm:controller_comparison}
For every closed-loop trajectory generated by the exact sparse OSAOC rule,
\begin{equation}
    \|P_{\mathcal T}e(k+1)\|_2
    \le
    \alpha_{\mathcal T}
    \|P_{\mathcal T}(F(k)e(k)+\eta(k))\|_2.
    \label{eq:target_component_contraction}
\end{equation}
Moreover, with
\[
    y(k)
    :=
    \begin{bmatrix}
    \|P_{\mathcal T}e(k)\|_2\\
    \|P_{\mathcal R}e(k)\|_2
    \end{bmatrix},
\]
one has the componentwise comparison
\begin{equation}
    y(k+1)\le \mathsf G(k)y(k)+q(k),
    \label{eq:controller_comparison}
\end{equation}
where
\begin{equation}
\begin{aligned}
    \mathsf G(k)
    &:={}
    \begin{bmatrix}
    \alpha_{\mathcal T}\|P_{\mathcal T}F(k)P_{\mathcal T}\|_2
    &
    \alpha_{\mathcal T}\|P_{\mathcal T}F(k)P_{\mathcal R}\|_2\\
    \|P_{\mathcal R}F(k)P_{\mathcal T}\|_2
    &
    \|P_{\mathcal R}F(k)P_{\mathcal R}\|_2
    \end{bmatrix},\\[1mm]
    q(k)
    &:={}
    \begin{bmatrix}
    \alpha_{\mathcal T}\|P_{\mathcal T}\eta(k)\|_2\\
    \|P_{\mathcal R}\eta(k)\|_2
    \end{bmatrix}.
\end{aligned}
\label{eq:comparison_Gq}
\end{equation}
Both \(\mathsf G(k)\) and \(q(k)\) depend only on the free dynamics,
the target set, the goal, and \(\gamma\); in particular, they are
independent of the realized support sequence and of the state trajectory.
\end{theorem}

\begin{proof}
Let
\[
    c_{\mathcal T}(k)
    :=P_{\mathcal T}(g-z(k))
    =-P_{\mathcal T}(F(k)e(k)+\eta(k)).
\]
For every \(i\in\mathcal T\), the singleton support \(\{i\}\) is
admissible.  Hence the optimal reduced support value in
\eqref{eq:set_objective} satisfies
\[
\begin{aligned}
    V(S^\star(k))
    &\ge
    \frac{\max_{i\in\mathcal T}c_i(k)^2}{1+\gamma}\\
    &\ge
    \frac{\|c_{\mathcal T}(k)\|_2^2}
         {n_{\mathcal T}(1+\gamma)}.
\end{aligned}
\]
By \eqref{eq:reduced_cost}, the optimized stage cost is
\[
    J_k^\star
    =
    \|c_{\mathcal T}(k)\|_2^2-V(S^\star(k)).
\]
Since \(J_k^\star\) is the post-control target-error energy plus the
nonnegative term \(\gamma u^\star(k)^2\),
\[
\begin{aligned}
    \|P_{\mathcal T}e(k+1)\|_2^2
    &\le J_k^\star\\
    &\le
    \left(1-\frac{1}{n_{\mathcal T}(1+\gamma)}\right)
    \|c_{\mathcal T}(k)\|_2^2,
\end{aligned}
\]
which proves \eqref{eq:target_component_contraction}.

Since every admissible support is contained in \(\mathcal T\),
\(P_{\mathcal R}b^\star(k)=0\).  Therefore
\[
    P_{\mathcal R}e(k+1)
    =P_{\mathcal R}(F(k)e(k)+\eta(k)).
\]
Decomposing \(e=P_{\mathcal T}e+P_{\mathcal R}e\), applying the triangle
inequality to this identity and to
\eqref{eq:target_component_contraction}, and collecting the four block
norms gives \eqref{eq:controller_comparison}--\eqref{eq:comparison_Gq}.
\end{proof}

\begin{corollary}[Checkable small-gain certificate]
\label{cor:apriori_small_gain}
Suppose there is a constant nonnegative matrix \(\overline{\mathsf G}\)
and a vector \(\bar q\ge0\) such that, componentwise,
\[
    \mathsf G(k)\le\overline{\mathsf G},
    \qquad
    q(k)\le\bar q,
    \qquad k\ge0.
\]
If
\begin{equation}
    \rho(\overline{\mathsf G})<1,
    \label{eq:uniform_small_gain_condition}
\end{equation}
then
\begin{equation}
    y(k)
    \le
    \overline{\mathsf G}^{\,k}y(0)
    +
    \sum_{j=0}^{k-1}\overline{\mathsf G}^{\,j}\bar q,
    \label{eq:uniform_small_gain_bound}
\end{equation}
and therefore
\begin{equation}
    \limsup_{k\to\infty}y(k)
    \le
    (I-\overline{\mathsf G})^{-1}\bar q
    \label{eq:uniform_small_gain_uub}
\end{equation}
componentwise.  If the target is affine invariant, so
\(\eta(k)\equiv0\), then \(e(k)\to0\) exponentially.

If \(F(k)\) and \(h(k)\) are periodic with period \(N_c\), a less
conservative check is obtained by defining
\begin{equation}
    \mathsf G_{\rm c}
    :=
    \mathsf G(N_c-1)\cdots\mathsf G(0).
    \label{eq:comparison_cycle_matrix}
\end{equation}
The condition \(\rho(\mathsf G_{\rm c})<1\) likewise gives a global
cycle-instant ultimate bound, and exact exponential tracking when
\(\eta(k)\equiv0\).
\end{corollary}

\begin{proof}
The first statement follows by iterating
\eqref{eq:controller_comparison} and using monotonicity of multiplication
by nonnegative matrices.  Under
\eqref{eq:uniform_small_gain_condition}, the Neumann series converges to
\((I-\overline{\mathsf G})^{-1}\).  In the periodic case, unrolling
\eqref{eq:controller_comparison} over one period produces a two-dimensional
affine recursion whose homogeneous matrix is
\(\mathsf G_{\rm c}\); the same argument applies at cycle instants.
\end{proof}

\begin{corollary}[Full-state certificate and controller-induced DeGroot tracking]
\label{cor:full_state_controller_certificate}
If \(\mathcal T=\{1,\ldots,n\}\), let
\begin{equation}
    \alpha_n
    :=
    \sqrt{1-\frac{1}{n(1+\gamma)}}.
    \label{eq:alpha_full_state}
\end{equation}
Then, for arbitrary time variation,
\begin{equation}
    \|e(k+1)\|_2
    \le
    \alpha_n\|F(k)\|_2\|e(k)\|_2
    +
    \alpha_n\|F(k)g+h(k)-g\|_2.
    \label{eq:full_state_one_step_bound}
\end{equation}
Consequently, if
\(\beta:=\sup_k\|F(k)\|_2<\infty\),
\(\varepsilon_g:=\sup_k\|F(k)g+h(k)-g\|_2<\infty\), and
\begin{equation}
    \alpha_n\beta<1,
    \label{eq:full_state_uniform_condition}
\end{equation}
then
\begin{equation}
    \|e(k)\|_2
    \le
    (\alpha_n\beta)^k\|e(0)\|_2
    +
    \alpha_n\varepsilon_g
    \frac{1-(\alpha_n\beta)^k}{1-\alpha_n\beta}.
    \label{eq:full_state_uniform_bound}
\end{equation}
For periodic free dynamics, the weaker cycle condition
\begin{equation}
    \alpha_n^{N_c}
    \prod_{t=0}^{N_c-1}\|F(t)\|_2<1
    \label{eq:full_state_periodic_condition}
\end{equation}
is sufficient.

In particular, if \(F(k)=W(k)\) is doubly stochastic and
\(g=\xi\mathbf1\), then \(\|W(k)\|_2=1\) and
\(W(k)g=g\).  Thus \eqref{eq:full_state_uniform_condition} holds because
\(\alpha_n<1\), and sparse OSAOC tracks the prescribed consensus value
globally and exponentially.  By contrast, the uncontrolled DeGroot error
map retains the unit consensus mode.  Hence this certificate can establish
convergence created by the sparse controller rather than inherited from
the free dynamics.
\end{corollary}

\begin{proof}
When \(P_{\mathcal T}=I\),
\eqref{eq:target_component_contraction} gives
\eqref{eq:full_state_one_step_bound}.  The uniform bound follows by
iteration.  Over one period, iteration of the homogeneous coefficient
gives the factor in \eqref{eq:full_state_periodic_condition}.  Finally,
for a doubly stochastic matrix,
\(\|W(k)\|_2\le\sqrt{\|W(k)\|_1\|W(k)\|_\infty}=1\), while
\(W(k)\mathbf1=\mathbf1\) implies \(\|W(k)\|_2\ge1\); hence
\(\|W(k)\|_2=1\).
\end{proof}

Assume from this point onward that the uncontrolled affine dynamics
are periodic with period \(N_c\ge1\):
\begin{equation}
    F(k+N_c)=F(k),\qquad h(k+N_c)=h(k),\qquad k\ge0.
    \label{eq:W_periodic}
\end{equation}
This includes periodic DeGroot dynamics and, with constant
\(\Lambda\) and \(s\), periodic FJ dynamics generated by a periodic
row-stochastic \(W(k)\)
\cite{ProskurnikovEtAl2017TimeVarying}.
Define
\begin{equation}
    \beta
    :=
    \max_{0\le t<N_c}\|F(t)\|_2,
    \qquad
    \varepsilon_g
    :=
    \max_{0\le t<N_c}\|F(t)g+h(t)-g\|_2.
    \label{eq:beta_epsilon}
\end{equation}
Both constants are finite. Lemma~\ref{lem:I_minus_phi} gives, for
every \(k\),
\begin{equation}
    \|A(k)\|_2\le\beta,
    \qquad
    \|d(k)\|_2\le\varepsilon_g.
    \label{eq:A_d_uniform_bounds}
\end{equation}

For \(q\in\mathbb N\), define the finite geometric sum
\begin{equation}
    G_q(\beta)
    :=
    \sum_{j=0}^{q-1}\beta^j,
    \qquad
    G_0(\beta):=0.
    \label{eq:Gq_definition}
\end{equation}
Equivalently,
\[
    G_q(\beta)
    =
    \begin{cases}
        (\beta^q-1)/(\beta-1), & \beta\neq1,\\[1mm]
        q, & \beta=1.
    \end{cases}
\]

\subsection{Cycle Dynamics}
\label{subsec:cycle_dynamics}

For \(k\ge j\), let
\begin{equation}
    \Psi(k,j)
    :=
    \begin{cases}
        A(k-1)A(k-2)\cdots A(j), & k>j,\\
        I, & k=j.
    \end{cases}
    \label{eq:transition_operator}
\end{equation}
For cycle index \(m\in\mathbb N\), define
\begin{equation}
    \mathcal A_m
    :=
    \Psi((m+1)N_c,mN_c)
    =
    A((m+1)N_c-1)\cdots A(mN_c),
    \label{eq:cycle_transition}
\end{equation}
and
\begin{equation}
    \mathcal \zeta_m
    :=
    \sum_{i=mN_c}^{(m+1)N_c-1}
    \Psi((m+1)N_c,i+1)d(i).
    \label{eq:cycle_mismatch}
\end{equation}
Then
\begin{equation}
    e((m+1)N_c)
    =
    \mathcal A_m e(mN_c)+\mathcal \zeta_m.
    \label{eq:cycle_error_recursion}
\end{equation}

Although the free dynamics are periodic, neither \(A(k)\) nor
\(\mathcal A_m\) is generally periodic because
\(b^\star(k)\) depends on \(x(k)\).

The mismatch term $\zeta_m$ in \eqref{eq:cycle_error_recursion} is uniformly
bounded without an additional assumption.

\begin{lemma}[Explicit cycle-mismatch bound]
\label{lem:cycle_mismatch_bound}
For every cycle \(m\),
\begin{equation}
    \|\mathcal \zeta_m\|_2
    \le
    \delta_g,
    \qquad
    \delta_g
    :=
    \varepsilon_g G_{N_c}(\beta).
    \label{eq:delta_g}
\end{equation}
\end{lemma}

\begin{proof}
By \eqref{eq:A_d_uniform_bounds}, a transition product containing
\(q\) factors satisfies
\(\|\Psi(k,j)\|_2\le\beta^q\). Hence
\[
\begin{aligned}
    \|\mathcal \zeta_m\|_2
    &\le
    \sum_{i=mN_c}^{(m+1)N_c-1}
    \|\Psi((m+1)N_c,i+1)\|_2\,\|d(i)\|_2\\
    &\le
    \varepsilon_g
    \sum_{q=0}^{N_c-1}\beta^q
    =
    \varepsilon_g G_{N_c}(\beta).
\end{aligned}
\]
\end{proof}

\subsection{Conditional Contraction Along the Closed-Loop Error}
\label{subsec:cycle_contraction}

The a priori condition in Corollary~\ref{cor:apriori_small_gain} is
checkable before simulation but can be conservative, especially when only
a proper subset of agents is directly targeted.  We therefore retain a
complementary trajectory-wise analysis for periodic free dynamics.  It is
strictly conditional: the hypothesis below concerns the realized
closed-loop error direction and is not claimed to follow from the model
structure alone.

For each cycle, define the homogeneous realized cycle gain
\begin{equation}
    \chi_m
    :=
    \begin{cases}
    \displaystyle
    \frac{\|\mathcal A_m e(mN_c)\|_2}{\|e(mN_c)\|_2},
        & e(mN_c)\neq0,\\[3mm]
    0,  & e(mN_c)=0.
    \end{cases}
    \label{eq:realized_cycle_gain}
\end{equation}

\begin{assumption}[Uniform contraction along cycle-start errors]
\label{ass:cycle_contraction}
There exists \(a\in[0,1)\) such that
\begin{equation}
    \|\mathcal A_m e(mN_c)\|_2
    \le
    a\|e(mN_c)\|_2
    \qquad
    \text{for every }m\ge0.
    \label{eq:cycle_contraction_assumption}
\end{equation}
Equivalently, \(\chi_m\le a<1\) for all \(m\).
\end{assumption}

\begin{remark}[Nature of the assumption]
\label{rem:conditional_stability}
Assumption~\ref{ass:cycle_contraction} is a condition on the
\emph{state-dependent closed-loop trajectory}, not on the induced norm
of \(\mathcal A_m\) in every direction. It is therefore strictly weaker
than requiring \(\|\mathcal A_m\|_2<1\). In particular, if
\(e(mN_c)=0\), condition \eqref{eq:cycle_contraction_assumption} is
automatically satisfied regardless of \(\|\mathcal A_m\|_2\). This is
compatible with the controller switching off after exact tracking has
been achieved.

If the same constant \(a<1\) is proved for every trajectory starting
in a specified forward-invariant set, the bounds below are uniform
over that set. If the condition is verified only along one simulated
trajectory, the conclusions are trajectory-wise and should be
interpreted as an a posteriori validation of the contraction
hypothesis.  In particular, Assumption~\ref{ass:cycle_contraction} can
fail for admissible initial conditions and targets; Theorems~\ref{thm:cycle_uub}
and~\ref{thm:exact_invariant_tracking} below must therefore be read as
conditional results.  The a priori alternative is
Corollary~\ref{cor:apriori_small_gain}.
\end{remark}

\begin{proposition}[A stronger active-cycle sufficient condition]
\label{prop:active_cycle_operator_condition}
If there exists \(a\in[0,1)\) such that
\begin{equation}
    \|\mathcal A_m\|_2\le a
    \qquad
    \text{for every cycle with }e(mN_c)\neq0,
    \label{eq:active_cycle_operator_condition}
\end{equation}
then Assumption~\ref{ass:cycle_contraction} holds.
\end{proposition}

\begin{proof}
For every cycle with \(e(mN_c)\neq0\),
\[
    \|\mathcal A_m e(mN_c)\|_2
    \le
    \|\mathcal A_m\|_2\,\|e(mN_c)\|_2
    \le
    a\|e(mN_c)\|_2.
\]
For \(e(mN_c)=0\), both sides of
\eqref{eq:cycle_contraction_assumption} are zero.
\end{proof}

The admissibility of the empty support prevents a strict
state-independent certificate based on maximizing the cycle norm over
\emph{all} support sequences.

\begin{proposition}[Free-dynamics inheritance certificate and the DeGroot obstruction]
\label{prop:no_support_uniform_certificate}
For \(t=0,\ldots,N_c-1\) and \(b\in\mathcal U_r\), let
\[
    A_t(b)
    :=
    \left(
        I-\frac{bb^\top}{\mathbf 1^\top b+\gamma}
    \right)F(t),
\]
and define
\[
    \bar a_{\mathrm{all}}
    :=
    \max_{b_0,\ldots,b_{N_c-1}\in\mathcal U_r}
    \left\|
        A_{N_c-1}(b_{N_c-1})\cdots A_0(b_0)
    \right\|_2.
\]
Then
\[
    \bar a_{\mathrm{all}}
    \le
    \prod_{t=0}^{N_c-1}\|F(t)\|_2.
\]
Consequently,
\[
    \prod_{t=0}^{N_c-1}\|F(t)\|_2<1
\]
is sufficient to guarantee that \emph{every} admissible sparse correction
preserves strict induced-norm contraction over one cycle.  This condition
is control independent: it already certifies contraction of the all-empty,
uncontrolled support sequence and therefore is an inheritance/robustness
certificate, not a certificate of stabilization created by sparse OSAOC.

In the DeGroot specialization \(F(t)=W(t)\), with each \(W(t)\)
row-stochastic and the empty support admissible, one instead has
\begin{equation}
    \bar a_{\mathrm{all}}\ge1.
    \label{eq:no_uniform_certificate}
\end{equation}
Hence a strict support-uniform certificate
\(\bar a_{\mathrm{all}}<1\) cannot hold for the full DeGroot admissible
family \(\mathcal U_r^{N_c}\).
\end{proposition}

\begin{proof}
For any admissible \(b\), the same scaled-projector argument as in
Lemma~\ref{lem:I_minus_phi} gives
\[
    \left\|
        I-\frac{bb^\top}{\mathbf 1^\top b+\gamma}
    \right\|_2\le1.
\]
Therefore
\[
    \|A_t(b)\|_2\le\|F(t)\|_2,
\]
and submultiplicativity yields
\[
    \bar a_{\mathrm{all}}
    \le
    \prod_{t=0}^{N_c-1}\|F(t)\|_2.
\]
This proves the affine sufficient condition.

For DeGroot dynamics, \(b=0\) belongs to \(\mathcal U_r\).  Choosing
\(b_0=\cdots=b_{N_c-1}=0\) gives
\[
    A_{N_c-1}(0)\cdots A_0(0)
    =
    W(N_c-1)\cdots W(0).
\]
The product is row-stochastic and therefore has eigenvalue one.  Hence
\[
    \bar a_{\mathrm{all}}
    \ge
    \|W(N_c-1)\cdots W(0)\|_2
    \ge
    \rho\!\left(W(N_c-1)\cdots W(0)\right)
    =1,
\]
which proves \eqref{eq:no_uniform_certificate}.
\end{proof}

Proposition~\ref{prop:no_support_uniform_certificate} should therefore be
read as a robustness statement: a sufficiently contractive free affine
cycle remains contractive under every admissible sparse correction.  It
does not demonstrate controller-induced stabilization.  This distinction
is structural in the DeGroot case.  Since every row-stochastic \(W(t)\)
has spectral radius one,
\[
    \|W(t)\|_2\ge1,
    \qquad
    \beta\ge1,
    \qquad
    G_{N_c}(\beta)\ge N_c,
\]
and, by \eqref{eq:no_uniform_certificate},
\(\bar a_{\mathrm{all}}\ge1\).  Hence this control-independent norm test
cannot provide an a priori controller-induced contraction certificate for
the full DeGroot support family.  In an FJ model
\(F(k)=\Lambda W(k)\), the all-empty cycle need not carry a unit mode, so
the inheritance test may hold when the free FJ dynamics are already
contractive.  Controller-induced convergence is instead addressed by
Corollaries~\ref{cor:apriori_small_gain} and
\ref{cor:full_state_controller_certificate}.

\begin{theorem}[Cycle-instant ultimate bound]
\label{thm:cycle_uub}
Along every closed-loop trajectory satisfying
Assumption~\ref{ass:cycle_contraction},
\begin{equation}
    \|e(mN_c)\|_2
    \le
    a^m\|e(0)\|_2
    +
    \delta_g\frac{1-a^m}{1-a},
    \qquad m\ge0,
    \label{eq:cycle_bound_finite_m}
\end{equation}
where \(\delta_g\) is given by \eqref{eq:delta_g}. Consequently,
\begin{equation}
    \limsup_{m\to\infty}\|e(mN_c)\|_2
    \le
    \frac{\delta_g}{1-a}
    =
    \frac{\varepsilon_g G_{N_c}(\beta)}{1-a}.
    \label{eq:cycle_uub}
\end{equation}
\end{theorem}

\begin{proof}
From \eqref{eq:cycle_error_recursion},
Assumption~\ref{ass:cycle_contraction}, and
Lemma~\ref{lem:cycle_mismatch_bound},
\[
\begin{aligned}
    \|e((m+1)N_c)\|_2
    &\le
    \|\mathcal A_m e(mN_c)\|_2+\|\mathcal \zeta_m\|_2\\
    &\le
    a\|e(mN_c)\|_2+\delta_g.
\end{aligned}
\]
Iterating this scalar recursion yields
\eqref{eq:cycle_bound_finite_m}; taking \(m\to\infty\) gives
\eqref{eq:cycle_uub}.
\end{proof}

The cycle bound also yields an all-time bound.

\begin{corollary}[All-time ultimate bound]
\label{cor:all_time_uub}
Let \(k=mN_c+\tau\) with
\(\tau\in\{0,\ldots,N_c-1\}\). Under
Assumption~\ref{ass:cycle_contraction},
\begin{equation}
\begin{aligned}
    \|e(mN_c+\tau)\|_2
    \le\;&
    \beta^\tau
    \left[
        a^m\|e(0)\|_2
        +
        \delta_g\frac{1-a^m}{1-a}
    \right]\\
    &+
    \varepsilon_g G_\tau(\beta).
\end{aligned}
\label{eq:all_time_bound}
\end{equation}
Hence, for each fixed phase \(\tau\),
\begin{equation}
    \limsup_{m\to\infty}
    \|e(mN_c+\tau)\|_2
    \le
    \beta^\tau\frac{\delta_g}{1-a}
    +
    \varepsilon_g G_\tau(\beta).
    \label{eq:all_time_uub_phase}
\end{equation}
\end{corollary}

\begin{proof}
Unrolling \eqref{eq:error_affine} from \(mN_c\) to \(mN_c+\tau\)
gives
\[
    e(mN_c+\tau)
    =
    \Psi(mN_c+\tau,mN_c)e(mN_c)
    +
    \sum_{i=mN_c}^{mN_c+\tau-1}
    \Psi(mN_c+\tau,i+1)d(i).
\]
Apply \eqref{eq:A_d_uniform_bounds} and
Theorem~\ref{thm:cycle_uub}.
\end{proof}

\subsection{Exact Tracking of Affine-Invariant Profiles}
\label{subsec:invariant_targets}

The ultimate bound collapses to zero when the desired profile is
invariant under every uncontrolled affine update.

\begin{theorem}[Exact tracking of an affine-invariant profile]
\label{thm:exact_invariant_tracking}
Suppose
\begin{equation}
    F(k)g+h(k)=g,
    \qquad
    k\ge0,
    \label{eq:target_invariance}
\end{equation}
and Assumption~\ref{ass:cycle_contraction} holds. Then
\begin{equation}
    \|e(mN_c)\|_2
    \le
    a^m\|e(0)\|_2,
    \label{eq:exact_cycle_decay}
\end{equation}
and, for \(k=mN_c+\tau\),
\begin{equation}
    \|e(k)\|_2
    \le
    \beta^\tau a^m\|e(0)\|_2.
    \label{eq:exact_all_time_decay}
\end{equation}
Consequently,
\begin{equation}
    \lim_{k\to\infty}x(k)=g.
    \label{eq:exact_tracking_limit}
\end{equation}
\end{theorem}

\begin{proof}
By the definition \eqref{eq:A_d_definitions}
\[
d(k)=(I-\Phi(k))\bigl(F(k)g+h(k)-g\bigr),
\]
so condition \eqref{eq:target_invariance}
implies \(d(k)\equiv0\), and therefore
\(\zeta_m\equiv0\).
Thus the cycle recursion \eqref{eq:cycle_error_recursion}
reduces to
\[
    e((m+1)N_c)=\mathcal A_m e(mN_c).
\]
Assumption~\ref{ass:cycle_contraction} therefore gives
\[
    \|e((m+1)N_c)\|_2
    \le
    a\|e(mN_c)\|_2,
\]
which yields \eqref{eq:exact_cycle_decay} by iteration. Between cycle
instants,
\[
    e(mN_c+\tau)=\Psi(mN_c+\tau,mN_c)e(mN_c),
\]
and \(\|\Psi(mN_c+\tau,mN_c)\|_2\le\beta^\tau\), proving
\eqref{eq:exact_all_time_decay}. Since
\(\tau\le N_c-1\) and \(a^m\to0\), the full trajectory converges to
\(g\).
\end{proof}

\begin{corollary}[Consensus targets in the DeGroot and Friedkin--Johnsen special cases]
\label{cor:consensus_targets}
Consider a constant consensus target
\begin{equation}
    g=\alpha\mathbf 1,
    \qquad \alpha\in\mathbb R.
    \label{eq:consensus_target}
\end{equation}

For the DeGroot specialization \(F(k)=W(k)\), \(h(k)=0\), every
row-stochastic \(W(k)\) satisfies \(W(k)\mathbf1=\mathbf1\).
Hence every target \eqref{eq:consensus_target} satisfies
\eqref{eq:target_invariance}.  Under
Assumption~\ref{ass:cycle_contraction}, sparse OSAOC therefore yields
\[
    \lim_{k\to\infty}x_i(k)=\alpha,
    \qquad i=1,\ldots,n.
\]

For the Friedkin--Johnsen specialization
\[
    F(k)=\Lambda W(k),
    \qquad
    h(k)=(I-\Lambda)s,
\]
with row-stochastic \(W(k)\), the same consensus target is affine
invariant if and only if
\[
    (I-\Lambda)(s-\alpha\mathbf1)=0.
\]
Equivalently, every agent with \(\lambda_i<1\) must have innate opinion
\(s_i=\alpha\); agents with \(\lambda_i=1\) impose no restriction
through the prejudice term.  Under this condition and
Assumption~\ref{ass:cycle_contraction}, sparse OSAOC again tracks
\(\alpha\mathbf1\) exactly.
\end{corollary}

\begin{proof}
The DeGroot statement follows immediately from row-stochasticity.  For
the FJ case,
\[
\begin{aligned}
    F(k)(\alpha\mathbf1)+h(k)-\alpha\mathbf1
    &=\Lambda W(k)(\alpha\mathbf1)
      +(I-\Lambda)s-\alpha\mathbf1\\
    &=(I-\Lambda)(s-\alpha\mathbf1),
\end{aligned}
\]
so \eqref{eq:target_invariance} is equivalent to the stated condition.
\end{proof}

More generally, a nonconsensus FJ target \(g\) is exactly invariant only
when
\[
    \Lambda W(k)g+(I-\Lambda)s=g
    \qquad\text{for every }k.
\]
If this identity fails, Theorem~\ref{thm:cycle_uub} provides an ultimate
tracking bound rather than an exact-tracking conclusion.  This is the
natural effect of prejudices that are incompatible with the desired
profile \cite{ParsegovEtAl2017,ProskurnikovEtAl2017TimeVarying}.

\begin{remark}[Why no separate non-target connectivity theorem is needed]
\label{rem:no_extra_connectivity}
The error \(e(k)=x(k)-g\) in
Theorem~\ref{thm:exact_invariant_tracking} is the full
\(n\)-dimensional network error. Thus the cycle-contraction assumption
already concerns the realized full-state error, including target and
non-target coordinates. Once \eqref{eq:exact_cycle_decay} holds,
bounded inter-cycle transition products are sufficient to obtain
\eqref{eq:exact_all_time_decay}; an additional joint-connectivity
assumption is not used. A genuinely graph-theoretic propagation theorem
would require a weaker hypothesis that establishes contraction only on
the actuated or target subsystem and then uses connectivity to
propagate that convergence to the remaining nodes.
\end{remark}

\section{Competitive Sparse OSAOC Game}
\label{sec:multiplayer}

We now consider \(M\ge2\) external players acting directly on the same
time-varying affine opinion network:
\begin{equation}
    x(k+1)=F(k)x(k)+h(k)+\sum_{m=1}^{M} b_m(k)u_m(k).
    \label{eq:multiplayer_dynamics}
\end{equation}
Player \(m\) has target set \(\mathcal T_m\), goal \(g_m\), cardinality
budget \(r_m\), and penalty \(\gamma_m>0\), with
\[
    \mathcal U_m:=\{b_m\in\{0,1\}^n:
       \operatorname{supp}(b_m)\subseteq\mathcal T_m,\;
       \mathbf1^\top b_m\le r_m\}.
\]
Target sets may overlap.  With the free response
\(z=F(k)x(k)+h(k)\), let \(P_m\) denote the diagonal projector onto
\(\mathcal T_m\).  The pure strategy of player \(m\) is
\(a_m=(b_m,u_m)\in\mathcal A_m:=\mathcal U_m\times\mathbb R\), and
\begin{equation}
    J_m(a;z)=
    \left\|P_m\!\left(z+\sum_{\ell=1}^{M}b_\ell u_\ell-g_m\right)\right\|_2^2
    +\gamma_m u_m^2.
    \label{eq:multiplayer_cost}
\end{equation}
We call this a \emph{binary--continuous stage game}: the support decision
is discrete and the scalar action is continuous; no randomized mixed
strategies are involved.  The best response is generally a correspondence,
\[
    \operatorname{BR}_m(a_{-m};z)
    :=\arg\min_{(b_m,u_m)\in\mathcal A_m}J_m((b_m,u_m),a_{-m};z),
\]
because support ties may occur.

\subsection{Exact Sparse Best Responses}

For fixed opponents, define
\[
    \widetilde c_m:=P_m\!\left(g_m-z-\sum_{\ell\ne m}b_\ell u_\ell\right).
\]
Since \(P_m b_m=b_m\), player \(m\)'s problem has exactly the
single-player form
\[
    J_m=\|\widetilde c_m-b_m u_m\|_2^2+\gamma_m u_m^2.
\]
Hence, for fixed \(b_m\),
\begin{equation}
    u_m^\star
    =\frac{b_m^\top\widetilde c_m}{\mathbf1^\top b_m+\gamma_m},
    \label{eq:multi_fixed_support_action}
\end{equation}
and support selection reduces to
\[
    \max_{b_m\in\mathcal U_m}
    \frac{(b_m^\top\widetilde c_m)^2}{\mathbf1^\top b_m+\gamma_m}.
\]

\begin{proposition}[Exact sparse best response]
\label{prop:exact_sparse_BR}
For fixed opponent actions, a global best response of player \(m\) is
obtained by applying Theorem~\ref{thm:sorting} to the components of
\(\widetilde c_m\) on \(\mathcal T_m\).  Its complexity is
\(O(|\mathcal T_m|\log|\mathcal T_m|)\) by a full sort, or \(O(|\mathcal T_m|+r_m\log r_m)\) by partial extreme selection.
\end{proposition}

\subsection{Binary--Continuous Potential and Equilibrium Existence}

For \(b=(b_1,\ldots,b_M)\), define
\[
    B=[\,b_1\ \cdots\ b_M\,],\;
    u=[\,u_1\ \cdots\ u_M\,]^\top,\;
    \Gamma=\operatorname{diag}(\gamma_1,\ldots,\gamma_M).
\]
The fixed-direction continuous game and its quadratic potential are
analyzed in detail in \cite{GentilBhaya2026CompetitiveFJ}.  The following
result is the support-changing extension needed here.

\begin{theorem}[Exact potential for the sparse stage game]
\label{thm:mixed_exact_potential}
For every fixed \(z\), the binary--continuous game is an exact potential
game with
\begin{equation}
    \mathcal P(b,u;z)
    =\|z+Bu\|_2^2
    -2\sum_{m=1}^{M}u_m b_m^\top g_m
    +u^\top\Gamma u.
    \label{eq:mixed_potential}
\end{equation}
\end{theorem}

\begin{proof}
Let \(y=z+\sum_\ell b_\ell u_\ell\) and consider a unilateral change of
player \(m\), with \(\delta=b'_m u'_m-b_m u_m\).  Both supports lie in
\(\mathcal T_m\), so \(P_m\delta=\delta\).  Direct expansion gives
\[
    \Delta J_m
    =2\delta^\top(y-g_m)+\|\delta\|_2^2
      +\gamma_m[(u'_m)^2-u_m^2]
    =\Delta\mathcal P,
\]
which proves exact potentiality for simultaneous support-and-action
choices.
\end{proof}

A useful structural consequence is that the projectors \(P_m\) do not
appear explicitly in \eqref{eq:mixed_potential}. Target-set overlap changes
the admissible support families, but once a unilateral injection \(\delta\)
is supported on \(\mathcal T_m\), the identity \(P_m\delta=\delta\) removes
\(P_m\) from the potential increment. This is why arbitrary target-set
overlap is compatible with exact potentiality.

\begin{definition}[Instantaneous best-response equilibrium]
\label{def:IBRE}
At network time \(k\), a profile \(a^\star(k)\) is an
\emph{instantaneous best-response equilibrium} (IBRE) if it is a
pure-strategy Nash equilibrium of the frozen stage game, i.e.,
\[
    J_m(a_m^\star,a_{-m}^\star;z(k))
    \le J_m(a_m,a_{-m}^\star;z(k))
    \quad\forall a_m\in\mathcal A_m,\ \forall m.
\]
\end{definition}

\begin{theorem}[Existence of a pure-strategy IBRE]
\label{thm:IBRE_existence}
For every network time and state, the competitive sparse OSAOC stage game
admits at least one pure-strategy IBRE.
\end{theorem}

\begin{proof}
For every fixed support profile, \(B^\top B+\Gamma\succ0\), so the
continuous potential is coercive and has a unique minimizer; this is the
fixed-direction SPD structure of \cite{GentilBhaya2026CompetitiveFJ}.
There are finitely many support profiles, hence \(\mathcal P\) has a global
minimum over the full binary--continuous pure-strategy space.  Exact
potentiality makes every such minimizer a Nash equilibrium.
\end{proof}

\subsection{Fixed Supports and Implementation}
\label{subsec:fixed_support_subgame}

Fix a support profile and define
\begin{equation}
    G:=B^\top B,\qquad
    v_m:=b_m^\top g_m,\qquad
    S:=G+\Gamma\succ0.
    \label{eq:S_definition}
\end{equation}
The continuous subgame is exactly the fixed-influence game studied in
\cite{GentilBhaya2026CompetitiveFJ}.  In the present notation we have the
following immediate specialization.

\begin{proposition}[Unique fixed-support Nash equilibrium]
\label{prop:fixed_support_NE}
\begin{equation}
    u^{\mathrm{NE}}(b;z)=S^{-1}(v-B^\top z).
    \label{eq:fixed_support_NE}
\end{equation}
\end{proposition}
A fixed-support equilibrium is an IBRE only if no player can also improve
by changing its support.

For completeness, the fixed-support iterative facts needed to interpret
implementations are recalled from \cite{GentilBhaya2026CompetitiveFJ}.
Writing \(S=D+L+L^\top\), sequential scalar best responses are the
Gauss--Seidel iteration and converge for every initialization.  Parallel
scalar best responses are the Jacobi iteration and converge if and only if
\[
    2D-S\succ0.
\]
For two players this condition always holds when \(\Gamma\succ0\):
writing
\[
S=\begin{bmatrix}s_{11}&s_{12}\\s_{12}&s_{22}\end{bmatrix}\succ0,
\qquad
2D-S=\begin{bmatrix}s_{11}&-s_{12}\\-s_{12}&s_{22}\end{bmatrix},
\]
the two matrices have the same positive principal minors. No analogous
automatic implication holds for three or more players, for which the
Jacobi condition must be checked. These are \emph{frozen-support,
frozen-state} statements; they do not establish convergence of simultaneous
support-changing best responses or stability when the network state
evolves after one sweep.

When supports may change, a sequential sweep simply applies
Proposition~\ref{prop:exact_sparse_BR} player by player using the newest
available opponents' actions.  Each unilateral update cannot increase the potential, but the strategy
profile obtained after one sweep need not be an IBRE.  Likewise, exact
potentiality alone does not guarantee convergence of parallel
support-and-action updates.  The numerical section therefore treats
zero-opponent myopic control, one sequential sparse best-response sweep,
and exact IBRE computation as distinct decision rules.

\section{Competitive Welfare Loss}
\label{sec:price_competition}

Competitive loss is evaluated at the \emph{same} frozen free-response
state \(z\), with the same goals and admissible support families.  This
avoids conflating noncooperative inefficiency with different closed-loop
trajectories or with residual error that is unavoidable under sparse
actuation.  The fixed-influence full-state counterpart is analyzed in
\cite{GentilBhaya2026CompetitiveFJ}; here target-restricted costs and
strategic support selection require two additional layers.

\subsection{Same-Support Action Loss}

Fix \(b=(b_1,\ldots,b_M)\), set \(B=[\,b_1\ \cdots\ b_M\,]\), and let
\(P_m\) project onto \(\mathcal T_m\).  Define
\[
    D_{\mathcal T}:=\sum_{m=1}^{M}P_m,\qquad
    g_\Sigma:=\sum_{m=1}^{M}P_m g_m,
\]
and the total one-step social cost
\[
    J_\Sigma(u;b,z):=\sum_{m=1}^{M}
       \bigl[\|P_m(z+Bu-g_m)\|_2^2+\gamma_m u_m^2\bigr].
\]
Here ``social cost'' means the sum of the players' stated one-step
objectives, including their quadratic effort penalties. Equivalently, the
centralized benchmark uses the same effort weights \(\gamma_m\) as the
players. A planner with different effort valuations can replace \(\Gamma\)
by its own positive diagonal weight matrix in the centralized problem; we
retain the aggregate-player-cost convention so that the welfare gap
isolates noncooperative play under a common accounting of tracking error
and effort.

Direct expansion gives
\begin{align*}
     J_\Sigma(u;b,z)
     & =u^\top H(b)u
     +2u^\top B^\top(D_{\mathcal T}z-g_\Sigma)
     +c_\Sigma(z),\\
    H(b) &:=B^\top D_{\mathcal T}B+\Gamma\succ0.
\end{align*}
Consequently the unique same-support social action is
\begin{equation}
    u^{\mathrm{SO}}(b;z)
    =H(b)^{-1}B^\top(g_\Sigma-D_{\mathcal T}z).
    \label{eq:same_support_social_optimum}
\end{equation}
The Nash action is \eqref{eq:fixed_support_NE}.  With
\(e_{\mathrm{comp}}:=u^{\mathrm{NE}}-u^{\mathrm{SO}}\), completing the
square yields
\begin{equation}
    \Delta_{\mathrm{act}}(b;z)
    :=J_\Sigma(u^{\mathrm{NE}};b,z)-J_\Sigma(u^{\mathrm{SO}};b,z)
    =e_{\mathrm{comp}}^\top H(b)e_{\mathrm{comp}}\ge0.
    \label{eq:additive_action_loss}
\end{equation}
This is the target-restricted analogue of the same-state fixed-direction
welfare identity in \cite{GentilBhaya2026CompetitiveFJ}.  When the social
optimum is positive, the corresponding multiplicative ratio is
\(1+\Delta_{\mathrm{act}}/J_\Sigma(u^{\mathrm{SO}};b,z)\); the additive
quantity is used below because it remains well defined when the optimal
cost is zero.

\subsection{Support Geometry and Cross-Target Exposure}

The Gram entry \(G_{m\ell}=b_m^\top b_\ell\) measures co-location of
active supports but is not a complete welfare-externality measure. To make
the direction of exposure explicit, define
\begin{equation}
    C_{m\to\ell}(b):=b_m^\top P_\ell b_m
    =|\operatorname{supp}(b_m)\cap\mathcal T_\ell|,
    \label{eq:target_exposure_matrix}
\end{equation}
where the arrow reads ``action of player \(m\) exposed to player
\(\ell\)'s target set.'' Thus \(G_{m\ell}=0\) may coexist with
\(C_{m\to\ell}>0\): two players can act on different nodes even though
one player's active node enters the other player's objective.

If
\begin{equation}
    P_\ell b_m=0\qquad(m\ne\ell),
    \label{eq:externality_free_condition}
\end{equation}
then \(D_{\mathcal T}B=B\) and \(B^\top g_\Sigma=v\), so
\(H=S\), \(u^{\mathrm{SO}}=u^{\mathrm{NE}}\), and
\(\Delta_{\mathrm{act}}=0\) for every \(z\).  Disjoint target sets are a
sufficient special case.  This is a one-step statement only: an action may
still propagate through later affine updates and affect another player's
future target state.

\subsection{Full Sparse Competitive Loss}
\label{subsec:full_sparse_PoC}

Since the support is strategic, the action loss does not capture all
noncooperative inefficiency.  For each admissible support profile define
\[
    \widehat J_\Sigma(b;z):=
    J_\Sigma(u^{\mathrm{SO}}(b;z);b,z),
\]
and choose a centralized sparse support
\[
    b^{\mathrm{SO}}\in\arg\min_b\widehat J_\Sigma(b;z),
\]
where each \(b_m\in\mathcal U_m\).  Let
\(a^{\mathrm{NE}}=(b^{\mathrm{NE}},u^{\mathrm{NE}})\) be any IBRE.  The
full additive sparse competitive loss is
\begin{equation}
    \Delta_{\mathrm{sp}}
    :=J_\Sigma(u^{\mathrm{NE}};b^{\mathrm{NE}},z)
      -\widehat J_\Sigma(b^{\mathrm{SO}};z)\ge0.
    \label{eq:full_sparse_loss}
\end{equation}
If several IBREs exist, this value may depend on equilibrium selection.

\begin{theorem}[Action--support decomposition of competitive loss]
\label{thm:action_support_decomposition}
For every IBRE,
\begin{equation}
    \Delta_{\mathrm{sp}}
    =\Delta_{\mathrm{act}}+\Delta_{\mathrm{sup}},
    \qquad
    \Delta_{\mathrm{sup}}
    :=\widehat J_\Sigma(b^{\mathrm{NE}};z)
      -\widehat J_\Sigma(b^{\mathrm{SO}};z)\ge0,
    \label{eq:loss_decomposition}
\end{equation}
where
\(\Delta_{\mathrm{act}}=\Delta_{\mathrm{act}}(b^{\mathrm{NE}};z)\ge0\).
\end{theorem}

\begin{proof}
Add and subtract \(\widehat J_\Sigma(b^{\mathrm{NE}};z)\) in
\eqref{eq:full_sparse_loss}.  The two resulting differences are
nonnegative by \eqref{eq:additive_action_loss} and by optimality of
\(b^{\mathrm{SO}}\), respectively.
\end{proof}

Thus \(\Delta_{\mathrm{act}}\) measures selfish scalar actions conditional
on the equilibrium support, whereas \(\Delta_{\mathrm{sup}}\) measures the
additional loss from noncooperative support selection.  The residual cost
\(\widehat J_\Sigma(b^{\mathrm{SO}};z)\) is the best value achievable
within the prescribed sparse actuation class and is therefore a structural
feasibility limitation, not competitive loss.

\section{Numerical Experiments}
\label{sec:numerics}

\input{strategic_switching_outputs/RESULTTODO_values.tex}

\providecommand{\FJBeta}{0.840387}
\providecommand{\FJUniformCycleBound}{0.328079}
\providecommand{\FJEpsilonInvariant}{0.000000}
\providecommand{\FJEpsilonMismatch}{0.167070}
\providecommand{\FJAObsInvariant}{0.003382}
\providecommand{\FJAObsMismatch}{0.020028}
\providecommand{\FJUltimateInvariant}{0.000000}
\providecommand{\FJUltimateMismatch}{0.904809}
\providecommand{\FJFinalCycleInvariant}{0.000000}
\providecommand{\FJFinalCycleMismatch}{0.172672}
\providecommand{\FJMaxTargetInvariant}{0.000000}
\providecommand{\FJMaxTargetMismatch}{0.003023}
\providecommand{\FJLastActiveInvariant}{22}
\providecommand{\FJLastActiveMismatch}{throughout}
\providecommand{\FJBoundConservatism}{5.24}
\providecommand{\ActivationValueTol}{10^{-24}}
\providecommand{\CycleZeroTol}{10^{-9}}
\providecommand{\BudgetHorizon}{250}
\providecommand{\BudgetMCSamples}{50}
\providecommand{\CompetitiveMCSamples}{10}

The numerical study has five objectives. First, it verifies the exact
support-selection theorem against exhaustive enumeration and separates the
full-sort implementation cost from the cost of enumeration. Second, it uses a five-node DeGroot periodic-switching benchmark to evaluate
the trajectory-wise diagnostics of Section~III. Third, it applies the same
five switching matrices to invariant and prejudice-mismatched FJ dynamics. Fourth, it studies budget and
regularization effects, now supplemented by a fixed-seed Monte-Carlo sweep
over initial conditions and by baselines that separate adaptive support
choice, adaptive cardinality, and the cardinality constraint itself.
Finally, the multiplayer experiment distinguishes several stage-game
implementations and supplements the nominal trajectory by a small sweep
over independently generated periodic networks.

Unless stated otherwise, all simulations use double-precision arithmetic.
For the closed-loop simulations, Algorithm~\ref{alg:sparse_osaoc} uses the
objective-improvement dead-band
\[
    \tau_{\rm act}=\ActivationValueTol .
\]
This value is many orders of magnitude below the reported one-step
improvements and is used only to prevent floating-point round-off from
creating formally nonempty supports after the action has become numerically
zero; the analytical results correspond to \(\tau_{\rm act}=0\).

\subsection{Performance Metrics}
\label{subsec:numerical_metrics}

For a single-player target set \(\mathcal T\), define
\begin{equation}
    E_{\mathcal T}(k)
    :=
    \left\|P_{\mathcal T}(x(k)-g)\right\|_2^2
    \label{eq:numerical_target_error}
\end{equation}
and the cumulative control energy
\begin{equation}
    \mathcal E_u(K)
    :=
    \sum_{k=0}^{K-1}u(k)^2.
    \label{eq:numerical_control_energy}
\end{equation}
The instantaneous actuation cardinality is
\begin{equation}
    s(k):=\|b^\star(k)\|_1.
    \label{eq:numerical_cardinality}
\end{equation}

For a simulation ending at \(K\), we use the finite-horizon settling time
\begin{equation}
\begin{aligned}
    t_{\mathrm{set}}(\epsilon;K)
    :=
    \min\Bigl\{
        k:\;
        E_{\mathcal T}(j)\le\epsilon
        \text{ for every }j=k,\ldots,K
    \Bigr\}.
\end{aligned}
\label{eq:settling_time_definition}
\end{equation}
If the set is empty, the trajectory is reported as not settled within the
horizon. Thus \(t_{\rm set}\) is not merely a persistence-window statistic.

For periodic networks, ratios involving a numerically zero cycle-start
error are excluded. Define
\[
\begin{aligned}
    \mathcal M_K^{\rm nz}
    &:=\{m:\ mN_c,(m+1)N_c\le K,\
             \ \|e(mN_c)\|_2>\tau_{\rm cyc}\},\\
    \tau_{\rm cyc}&=\CycleZeroTol .
\end{aligned}
\]
The reported realized diagnostic is
\begin{equation}
    a_{\mathrm{obs}}
    :=
    \max_{m\in\mathcal M_K^{\rm nz}}\chi_m,
    \label{eq:observed_cycle_contraction}
\end{equation}
where \(\chi_m\) is the homogeneous cycle gain defined in
\eqref{eq:realized_cycle_gain}. This threshold avoids ratios dominated by
machine round-off and affects the diagnostic only, not the simulated
trajectory.

\subsection{Exactness and Computational Cost of Sparse Selection}
\label{subsec:sorting_validation}

The first experiment isolates support selection from the network dynamics.
For randomly generated residual vectors \(c\in\mathbb R^N\), budgets
\(r\le N\), and penalties \(\gamma>0\),
Algorithm~\ref{alg:sparse_osaoc} is compared with exhaustive enumeration of
all admissible supports. For each instance, record
\begin{equation}
    \varepsilon_V
    :=
    \left|
        V(S_{\mathrm{sort}})
        -
        V(S_{\mathrm{enum}})
    \right|.
    \label{eq:objective_validation_error}
\end{equation}
Since optimal supports need not be unique, equality of objective values is
the appropriate correctness test. The validation set includes mixed-sign
and deliberately tied residuals. Across \NumSortingTests{} instances, the
maximum discrepancy was \(\MaxSortingError\).

For Table~\ref{tab:sorting_validation}, \(r=5\), \(\gamma=0.5\), and the
residual entries are i.i.d. standard normal. Each row reports the median of
20 wall-clock repetitions on identical instances for the two methods; the
complete software/hardware record is exported with the numerical archive.

\input{strategic_switching_outputs/table_sorting_validation.tex}

For fixed \(r\), exhaustive enumeration visits
\(\sum_{s=0}^{r}\binom Ns=\Theta(N^r)\) supports and is therefore
polynomial in \(N\), though with a rapidly growing degree. It becomes
exponential when \(r\) scales with \(N\). By contrast, the full-sort
implementation of Theorem~\ref{thm:sorting} is \(O(N\log N)\), and the
partial-selection implementation discussed in Section~\ref{sec:single_player}
is \(O(N+r\log r)\).

\subsection{Sparse Tracking under Periodic Switching}
\label{subsec:periodic_tracking_numerics}

We adopt the five representative stochastic transition matrices from
Example~2 of Wang et al. \cite{WangEtAl2023ConcatenatedFJ}. In the source
construction, a discussion event contains two stubborn participants, one
nonstubborn participant, and two absent agents. For the ordered three-agent
participating block, the coarse-grained and representative matrices are
\[
 \bar P_{\rm part}=
 \begin{bmatrix}
 3/4&1/4&0\\[1mm]1/4&3/4&0\\[1mm]1/2&1/2&0
 \end{bmatrix},
 \qquad
 P_{\rm part}=
 \begin{bmatrix}
 1/2&1/2&0\\[1mm]1/2&1/2&0\\[1mm]1/2&1/2&0
 \end{bmatrix}.
\]
The five events are
\[
 (12|3|45),\ (34|5|12),\ (51|2|34),\ (23|4|51),\ (45|1|23),
\]
where the first pair denotes the stubborn participants, the middle entry
the nonstubborn participant, and the final pair the absent agents. The resulting representative matrices are shown in Figure \ref{fig:five_node_matrices} and are also exported to \path{W5_matrices.csv}

\input{strategic_switching_outputs/five_node_matrices_fragment.tex}

We relabel these representatives \(W_j:=P_j\), since here they are used
directly as DeGroot one-step update matrices. The two periodic orders are
\begin{equation}
    W_1\rightarrow W_2\rightarrow W_3\rightarrow W_4\rightarrow W_5,
    \label{eq:consensus_switching_sequence_num}
\end{equation}
and
\begin{equation}
    W_1\rightarrow W_5\rightarrow W_4\rightarrow W_3\rightarrow W_2.
    \label{eq:oscillatory_switching_sequence_num}
\end{equation}
Wang et al. show that the first order converges to consensus for every
initial state, whereas the reverse order is periodic for every initial
state. Our experiment is not a literal simulation of their original
two-timescale FJ process: it reuses the representative coarse-grained
stochastic matrices as a DeGroot switching benchmark and adds direct sparse
OSAOC.

The parameters are
\begin{equation}
\begin{aligned}
    x(0)&=[\,0.43,\;0.18,\;0.90,\;0.97,\;0.45\,]^\top,\\
    \mathcal T&=\{1,3,5\},\qquad
    g=0.2\,\mathbf 1,\\
    r&=3,\qquad \gamma=0.5.
\end{aligned}
\label{eq:wang_simulation_parameters}
\end{equation}
For the first switching order, the uncontrolled trajectory converges to
\(\OpenConsensusLimit\); for the reverse order it oscillates approximately
between \(\OpenOscMin\) and \(\OpenOscMax\). Sparse OSAOC tracks the target
in both cases. With the explicit dead-band above, the last active controls
occur at \(k=\LastActiveConsensus\) and \(k=\LastActiveOscillatory\),
respectively.

We additionally report
\begin{equation}
    E_{\mathcal T}(k),\qquad
    s(k),\qquad
    u(k),\qquad
    \chi_m.
    \label{eq:tracking_diagnostics}
\end{equation}
Since \(g=0.2\,\mathbf 1\) is invariant under every row-stochastic \(W(k)\),
\(\zeta_m=0\). Hence, for every retained nonzero cycle,
\begin{equation}
    \chi_m
    =
    \frac{\|e((m+1)N_c)\|_2}{\|e(mN_c)\|_2}.
    \label{eq:cycle_gain_from_errors}
\end{equation}
The maximum retained gains are
\begin{equation}
\begin{aligned}
    a_{\mathrm{obs}}^{\mathrm{cons}}&=\AObsConsensus,\\
    a_{\mathrm{obs}}^{\mathrm{osc}}&=\AObsOscillatory.
\end{aligned}
\label{eq:observed_cycle_values}
\end{equation}
These are trajectory-wise diagnostics, not state-independent certificates.
Figure \ref{fig:periodic_tracking} summarizes the corresponding state trajectories, target errors, control actions, support cardinalities, and retained cycle gains.
\begin{figure}[t]
\centering
\includegraphics[width=\columnwidth]{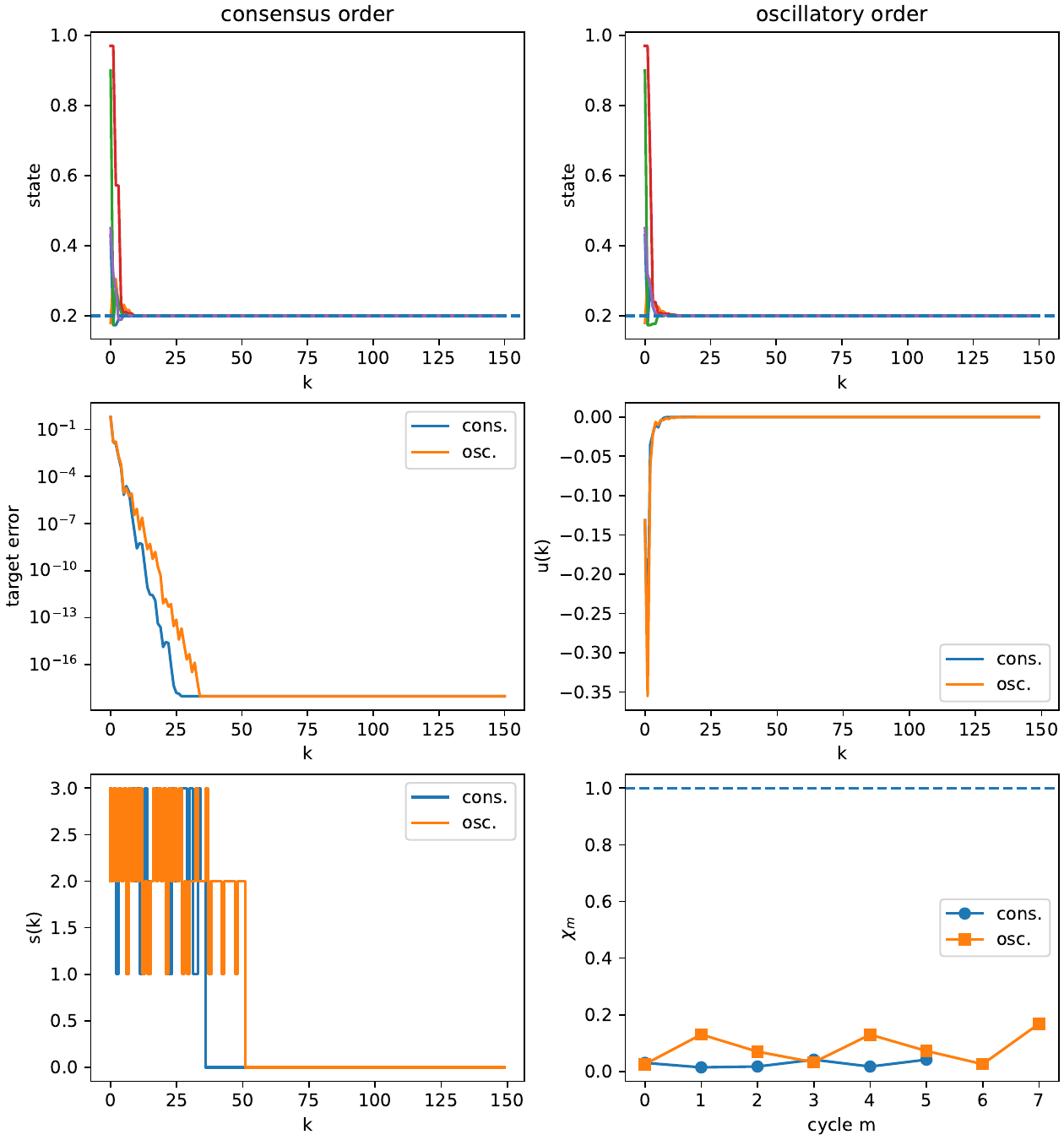}
\caption{Five-node periodic-switching example. State trajectories, target
error, control action, selected cardinality, and retained cycle gains are
shown for both switching orders.}
\label{fig:periodic_tracking}
\end{figure}

The DeGroot experiment tests the closed-loop tracking layer rather than the
exactness of the sorting theorem, which depends on the opinion model only
through the current residual. The all-empty support sequence still excludes
a strict support-uniform induced-norm certificate over the complete DeGroot
support family, as established in Section~\ref{sec:tracking}.

\subsection{Friedkin--Johnsen Specialization: Invariant and Mismatched Targets}
\label{subsec:fj_numerics}

To isolate the opinion-model effect from support selection, we reuse the
same five matrices, initial state, target set, consensus goal, budget, and
penalty, but use
\[
\begin{aligned}
 z(k)&=\Lambda W(k)x(k)+(I-\Lambda)s,\\
 \Lambda&=\operatorname{diag}(0.55,0.65,0.70,0.60,0.75).
\end{aligned}
\]
The intervention \(b(k)u(k)\) remains direct and is not premultiplied by
\(\Lambda\). The two prejudice vectors are
\[
 s_{\rm inv}=0.2\,\mathbf 1,
 \qquad
 s_{\rm mis}=[\,0.05,0.35,-0.10,0.45,0\,]^\top.
\]
The free-response matrices satisfy
\[
 \beta=\max_j\|F_j\|_2=\FJBeta,
 \qquad
 \prod_{j=1}^{5}\|F_j\|_2=\FJUniformCycleBound<1.
\]
As discussed in Section~\ref{sec:tracking}, this is an inheritance result:
the free FJ cycle is already contractive in this estimate, and sparse
correction does not destroy that contraction.

For \(s_{\rm inv}=g\), target invariance holds and
\(\varepsilon_g=\FJEpsilonInvariant\); the last active control occurs at
\(k=\FJLastActiveInvariant\). For \(s_{\rm mis}\), target invariance fails,
\(\varepsilon_g=\FJEpsilonMismatch\), and the controller remains active
throughout the horizon while approaching a period-five orbit. The bound in
Theorem~\ref{thm:cycle_uub} gives
\[
 \limsup_{m\to\infty}\|e(5m)\|_2\le\FJUltimateMismatch,
\]
whereas the final observed cycle-start error is
\(\FJFinalCycleMismatch\). Thus the a priori bound is about
\(\FJBoundConservatism\) times the observed value, quantifying its
conservatism.

In Table~\ref{tab:fj_specialization}, \(\chi_m\) always denotes the
\emph{homogeneous} gain
\(\|\mathcal A_m e(mN_c)\|_2/\|e(mN_c)\|_2\), as defined in
\eqref{eq:realized_cycle_gain}. In the mismatched affine case
\(\zeta_m\neq0\), so it is not the full ratio
\(\|e((m+1)N_c)\|_2/\|e(mN_c)\|_2\); the latter approaches one on the
periodic orbit.
Figure \ref{fig:fj_affine_tracking} compares the invariant and prejudice-mismatched FJ trajectories
\input{strategic_switching_outputs/table_fj_specialization.tex}

\begin{figure}[t]
\centering
\includegraphics[width=\columnwidth]{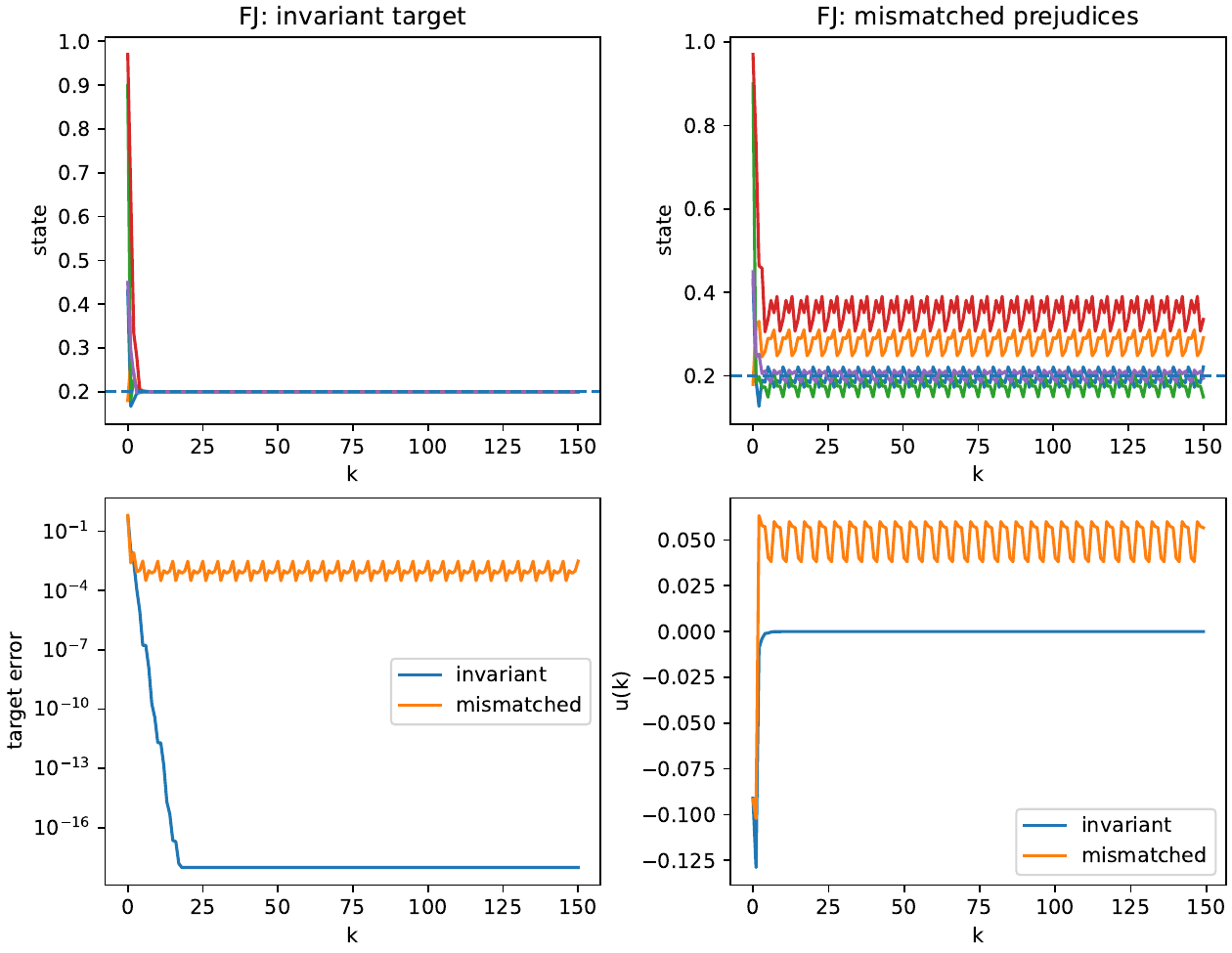}
\caption{Friedkin--Johnsen specialization. The invariant case tracks the
consensus target exactly; the mismatched case approaches a biased
period-five orbit and remains actively controlled.}
\label{fig:fj_affine_tracking}
\end{figure}

\subsection{Effect of the Actuation Budget}
\label{subsec:budget_numerics}

The deterministic eight-node benchmark uses
\begin{equation}
\begin{aligned}
 x(0)&=[\,-1.6,-1.5,-0.4,0.2,0.4,-0.1,-0.3,-1.2\,]^\top,\\
 \mathcal T&=\{3,4,7,8\},\qquad g=0.5\,\mathbf 1,\\
 r&\in\{1,2,3,4\},\qquad \gamma=0.1.
\end{aligned}
\label{eq:budget_simulation_parameters}
\end{equation}
The horizon and settling threshold are
\begin{equation}
    K=\BudgetHorizon,
    \qquad
    \epsilon=\SettleTol.
    \label{eq:settling_parameters}
\end{equation}
The mean cardinality in Table~\ref{tab:budget_results} is computed only for
\(k<t_{\rm set}\), so it is not dominated by inactive post-settling steps.

\input{strategic_switching_outputs/table_budget_results.tex}

\begin{figure}[t]
\centering
\includegraphics[width=\columnwidth]{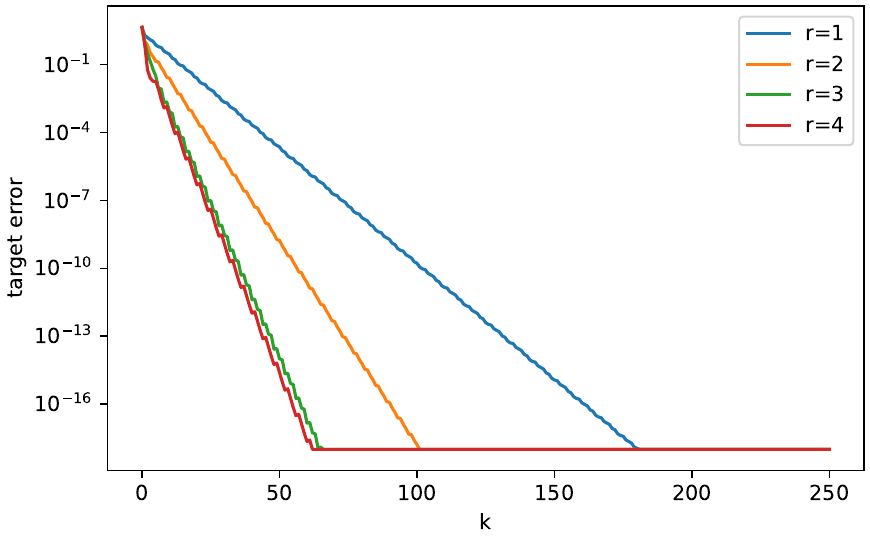}
\caption{Target-error histories for the four actuation budgets in the
nominal eight-node experiment.}
\label{fig:budget_target_error}
\end{figure}

To test whether the nominal trend is specific to one initial condition, we
repeat the experiment over \(\BudgetMCSamples\) fixed-seed initial states
with independent entries \(x_i(0)\sim\mathcal U[-1.5,1]\), holding the
network, target, \(\gamma\), and horizon fixed. All Monte-Carlo trajectories
settled within the horizon. Table~\ref{tab:budget_mc_results} reports the
median and interquartile range (IQR) of \(t_{\rm set}\) and the mean plus/minus one standard
deviation of cumulative target error and control energy.

\input{strategic_switching_outputs/table_budget_monte_carlo.tex}

Figure \ref{fig:budget_target_error} shows the corresponding target-error histories for the four cardinality budgets.
The enlarged support family guarantees only pointwise one-step
non-worsening as \(r\) increases; it does not imply a monotone settling time
along distinct closed-loop trajectories. The Monte-Carlo results show that
the faster settling observed in the nominal run is nevertheless robust for
this benchmark family.

We also sweep
\[
 \gamma\in\{\GammaGrid\}
\]
for \(r=2\), reporting the same finite-horizon settling time, cumulative
error and energy, and pre-settling mean cardinality.
Table \ref{tab:gamma_sweep_results} reports the resulting settling time, cumulative error, control energy, and mean pre-settling cardinality, while Figure \ref{fig:gamma_sweep} displays the associated error–energy tradeoff as \(\gamma\) varies.
\input{strategic_switching_outputs/table_gamma_sweep.tex}

\begin{figure}[t]
\centering
\includegraphics[width=\columnwidth]{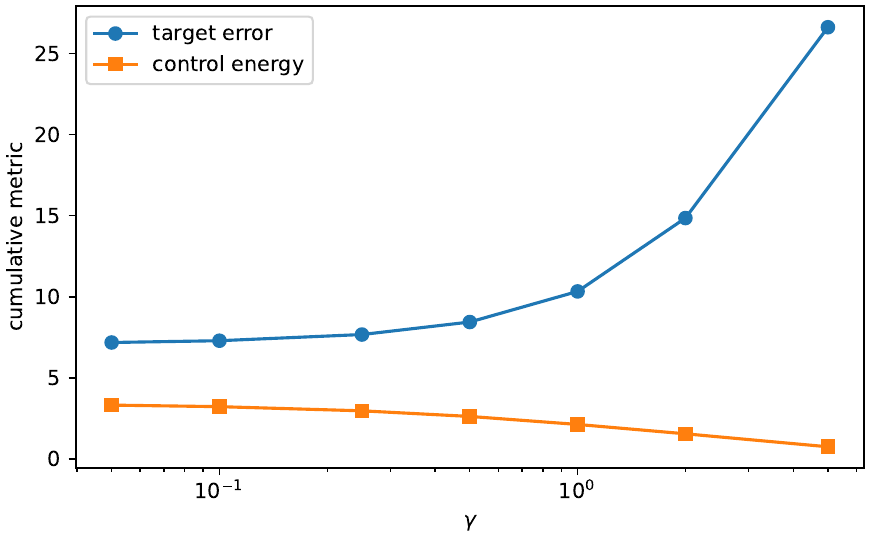}
\caption{Regularization sensitivity for \(r=2\): cumulative target error
and control energy as functions of \(\gamma\).}
\label{fig:gamma_sweep}
\end{figure}

\subsection{Competitive Sparse OSAOC}
\label{subsec:competitive_numerics}

The multiplayer benchmark uses
\begin{equation}
\begin{aligned}
    \mathcal T_1&=\{1,2,6,7,9\},\\
    \mathcal T_2&=\{3,4,5,7,10\},
\end{aligned}
\label{eq:corrected_multiplayer_targets}
\end{equation}
and
\begin{equation}
\begin{aligned}
    g_1&=0.5\,\mathbf 1,&r_1&=2,&\gamma_1&=0.5,\\
    g_2&=-0.5\,\mathbf 1,&r_2&=2,&\gamma_2&=0.5.
\end{aligned}
\label{eq:multiplayer_parameters}
\end{equation}
We compare zero-opponent myopic OSAOC, one sequential best-response sweep,
exact IBRE, and the centralized sparse optimum. Exact IBRE and centralized
optimization are computed by exhaustive support-profile enumeration only
because the benchmark is small.

For the exact IBRE we report
\begin{equation}
    G_{12}(k)=b_1(k)^\top b_2(k),
    \label{eq:numerical_G12}
\end{equation}
\begin{equation}
    C_{2\to1}(k)=b_2(k)^\top P_1b_2(k),
    \qquad
    C_{1\to2}(k)=b_1(k)^\top P_2b_1(k),
    \label{eq:numerical_exposure}
\end{equation}
and
\begin{equation}
    \Delta_{\rm act}(k),\qquad
    \Delta_{\rm sup}(k),\qquad
    \Delta_{\rm sp}(k)=\Delta_{\rm act}(k)+\Delta_{\rm sup}(k).
    \label{eq:numerical_welfare_metrics}
\end{equation}

\begin{figure*}[t]
\centering
\includegraphics[width=\textwidth]{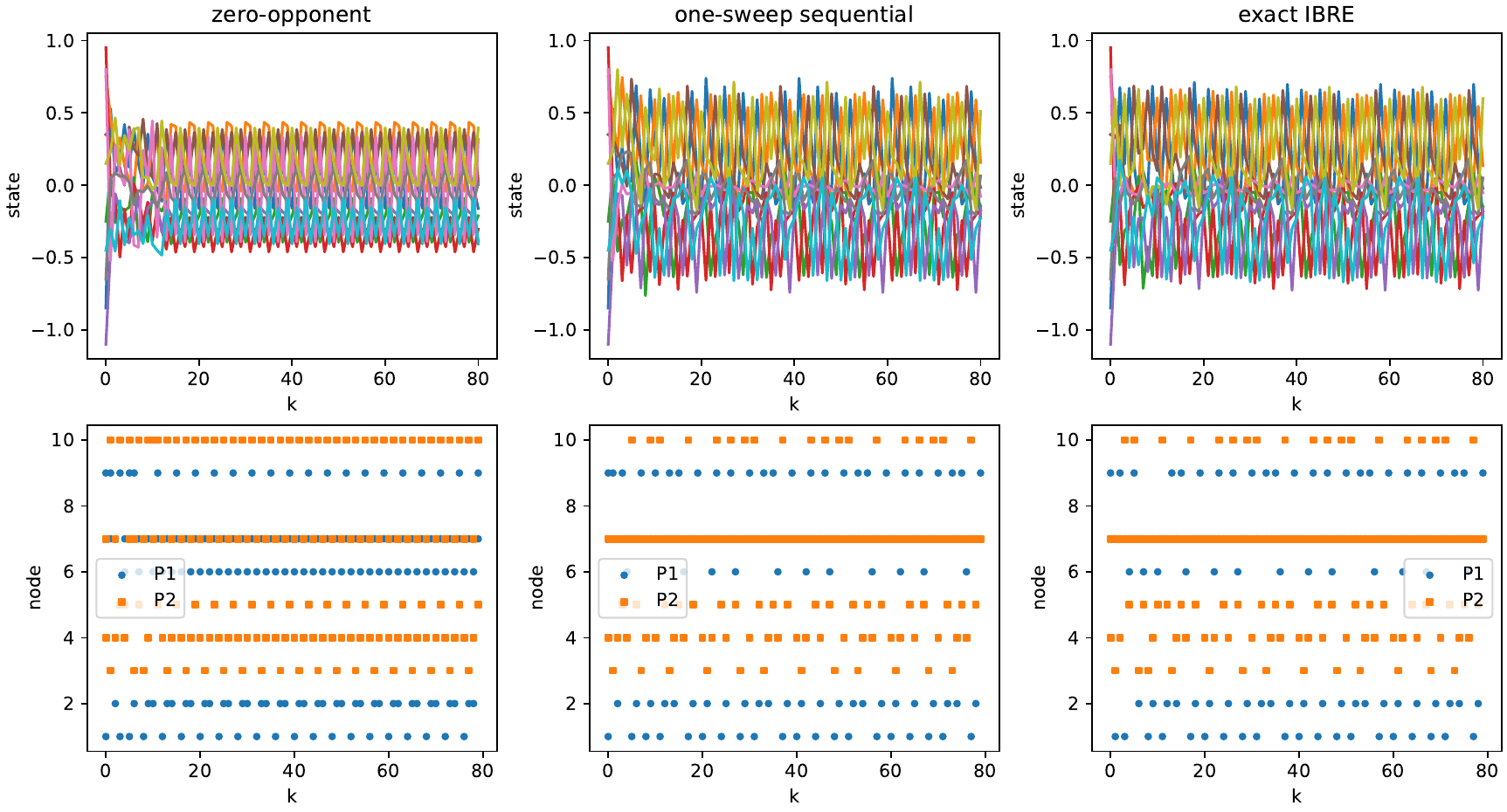}
\caption{Competitive sparse OSAOC under the zero-opponent, one-sweep
sequential, and exact-IBRE decentralized protocols.}
\label{fig:competitive_protocols}
\end{figure*}

\begin{figure}[t]
\centering
\includegraphics[width=\columnwidth]{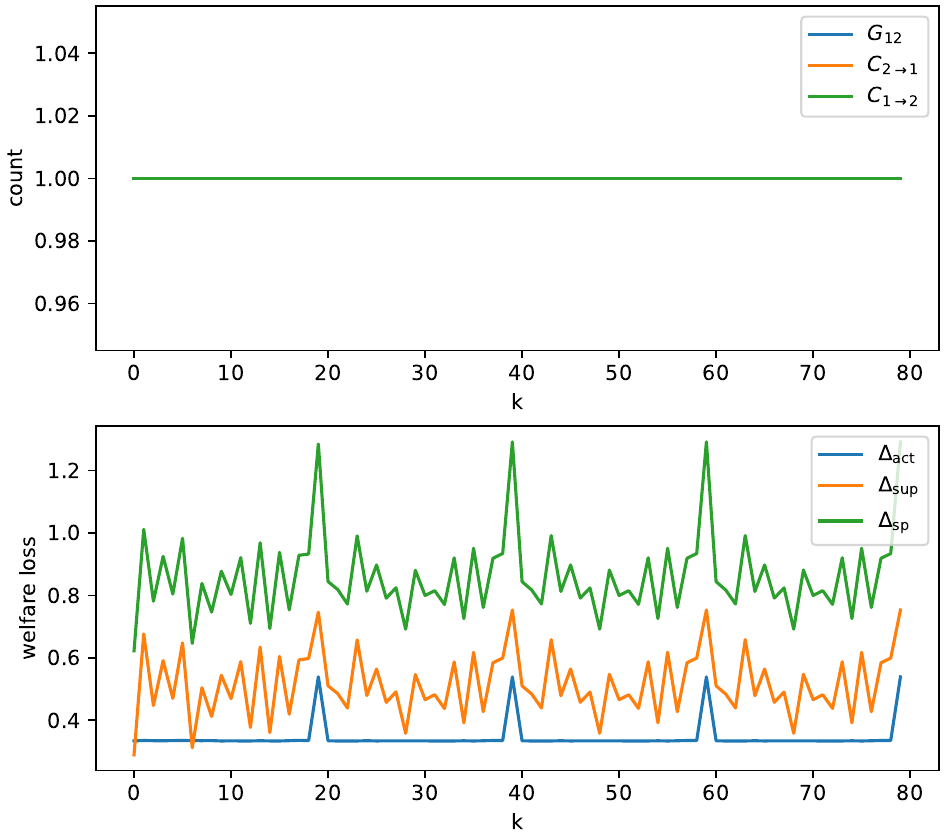}
\caption{Support overlap, directional cross-target exposure, and
same-state welfare loss along the exact-IBRE trajectory.}
\label{fig:competitive_welfare}
\end{figure}

Every computed IBRE is checked against unilateral deviations:
\begin{equation}
    \varepsilon_{\mathrm{NE}}
    :=
    \max_{k,m}
    \left[
        J_m(a^\star(k))-
        \min_{a_m}J_m(a_m,a^\star_{-m}(k))
    \right]_+
    =\MaxNEResidual.
    \label{eq:NE_validation_residual}
\end{equation}
Figure \ref{fig:competitive_protocols} compares the state and support trajectories of the three decentralized implementations.

Figure \ref{fig:competitive_welfare} shows the support-overlap, cross-target-exposure, and same-state welfare-loss diagnostics along the exact-IBRE trajectory.
\input{strategic_switching_outputs/table_competitive_summary.tex}

\paragraph{Equilibrium versus mutual myopia.}
For the nominal network, exact IBRE has a larger cumulative aggregate
player cost and substantially larger control energy than the zero-opponent
myopic rule. This does not contradict Nash optimality. IBRE excludes
profitable unilateral deviations at a \emph{fixed} free response against
the opponents' equilibrium actions; it neither minimizes aggregate cost
nor guarantees dominance over a different joint feedback policy that
generates a different state trajectory. The experiment therefore
illustrates that strategic consistency can itself be costly relative to
mutually nonstrategic behavior. The same ordering occurs in all
\(\CompetitiveMCSamples\) independently generated networks in the
sensitivity sweep below. We deliberately avoid calling the IBRE ``worse
for every player'': the reported quantity is aggregate cost, not a Pareto
comparison of the players' individual cumulative costs.

For the exact IBRE,
\begin{equation}
\begin{aligned}
    \sum_k\Delta_{\mathrm{act}}(k)&=\CumDeltaAct,\\
    \sum_k\Delta_{\mathrm{sup}}(k)&=\CumDeltaSup,\\
    \sum_k\Delta_{\mathrm{sp}}(k)&=\CumDeltaSp.
\end{aligned}
\label{eq:cumulative_welfare_decomposition}
\end{equation}
\input{strategic_switching_outputs/table_welfare_summary.tex}

Table~\ref{tab:competitive_summary} reports
each method's cumulative cost along that method's \emph{own} closed-loop
trajectory. In contrast, Table~\ref{tab:welfare_summary} sums same-state
counterfactual welfare gaps evaluated along the exact-IBRE trajectory: at
each frozen \(z(k)\), the IBRE is compared with centralized action and
support choices at that same state. Consequently
\(\sum_k\Delta_{\rm sp}(k)\) is not expected to equal the difference
between the exact-IBRE and centralized cumulative costs in
Table~\ref{tab:competitive_summary}.

As a modest network-sensitivity check, we repeat the zero-opponent, exact
IBRE, and centralized calculations on \(\CompetitiveMCSamples\)
independently generated four-phase periodic networks, with all player
parameters and the horizon fixed. The aggregate conclusions persist over
this sweep; in particular, the exact-IBRE cumulative social cost exceeds
the zero-opponent value in \(\MCIBREWorseCount\) of the
\(\CompetitiveMCSamples\) realizations. Table \ref{tab:competitive_mc} summarizes the resulting mean and standard deviation over the ten network realizations.

\input{strategic_switching_outputs/table_competitive_monte_carlo.tex}

These are stage-game and implementation experiments. They do not establish
a general boundedness or convergence result for the competitive closed
loop; the network sweep is empirical support for the reported strategic
phenomena.

\subsection{Comparison with Selection Baselines and the Budget Constraint}
\label{subsec:heuristic_baselines}

For the nominal eight-node experiment with \(r=2\), we compare:
\begin{enumerate}
 \item adaptive OSAOC with \(|S|\le r\);
 \item adaptive OSAOC constrained to \(|S|=r\) whenever the dead-band is
       exceeded;
 \item the support selected optimally at \(k=0\), then held fixed;
 \item a uniformly random exact-\(r\) support at every step, averaged over
       \(\NumRandomSeeds\) seeds; and
 \item adaptive OSAOC with the budget removed, \(r=|\mathcal T|\).
\end{enumerate}
The exact-\(r\) baseline isolates adaptive cardinality from adaptive node
choice, while the last baseline quantifies the closed-loop cost of imposing
\(r=2\) rather than allowing the full target set. Table \ref{tab:baseline_results} reports their cumulative target error, control energy, and settling time.

\input{strategic_switching_outputs/table_baselines.tex}

In this nominal run the \(|S|\le2\) and exact-\(2\) adaptive policies
coincide, so the improvement over static and random selection is due to
which nodes are selected, not to varying cardinality. Removing the budget
reduces both cumulative target error and control energy in this example,
showing that the cardinality restriction itself carries a measurable
closed-loop cost. Across the baselines, however, lower tracking error need
not coincide with lower energy; the comparison is therefore an
error--energy trade-off rather than a uniform dominance claim. Figure \ref{fig:baseline_error} visualizes these comparisons in the cumulative error–energy plane.

\begin{figure}[t]
\centering
\includegraphics[width=\columnwidth]{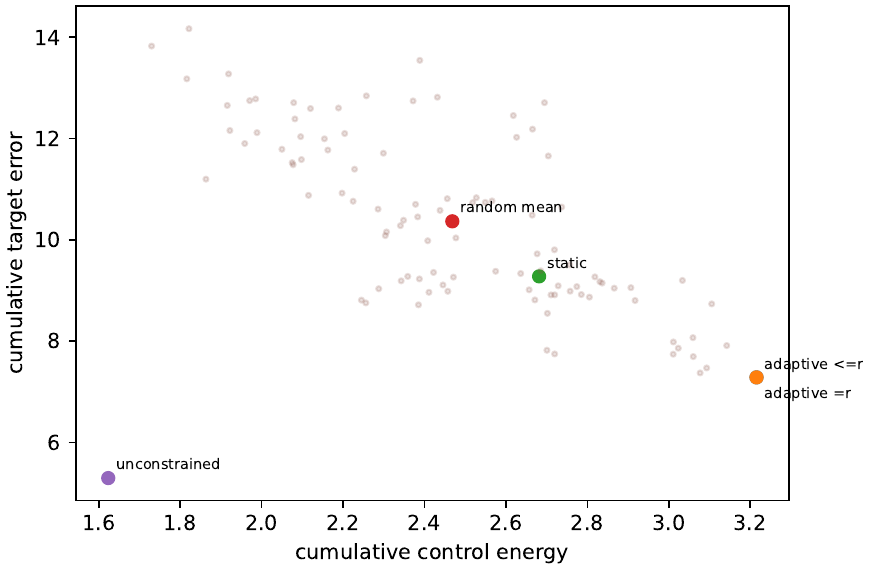}
\caption{Error--energy plane for adaptive, exact-cardinality, static,
random, and budget-unconstrained selection. Individual random-support runs
are shown together with their mean.}
\label{fig:baseline_error}
\end{figure}

\subsection{Reproducibility}
\label{subsec:reproducibility}

All experiments are generated by the Python script
\path{strategic_switching_experiments.py}, which exports every table
and figure appearing in this manuscript, the exact five-, eight-, and
ten-node benchmark matrices (\path{W5_matrices.csv},
\path{W8_matrices.csv}, and \path{W10_matrices.csv}), Monte-Carlo initial
conditions and per-run results, seeds, tolerances, software versions, and a
hardware/timing provenance record. The five-node matrices are also printed
in the paper; the larger matrices are supplied in machine-readable form and
in the supplementary network-matrix document. The complete Python source,
\path{strategic_switching_outputs/} directory, and supplementary files are
supplied with the submission. All files mentioned in this section are also available on the following \href{https://github.com/GabrielGent/Sparse-OSAOC.git}{GitHub}. ChatGPT and Claude were used to aid in debugging the code used in the numerical example.

\section{Conclusion}
\label{sec:conclusion}

Sparse OSAOC separates naturally into a model-independent one-step layer
and a model-dependent dynamical layer.  For any affine free response
\(z=Fx+h\), eliminating the scalar direct action leaves a support problem
solved globally by extreme-residual sorting; a full sort is
\(O(|\mathcal T|\log|\mathcal T|)\), while partial extreme selection
gives \(O(|\mathcal T|+r\log r)\).  More importantly for the
closed loop, one-step optimality also yields a controller-dependent target
contraction factor.  Combined with target/non-target block gains, this
gives an a priori small-gain certificate for global ultimate boundedness
and exact invariant-target tracking; in the full-state DeGroot case it can
certify convergence despite the unit mode of the free error dynamics.
Periodic cycle bounds remain useful as complementary, less restrictive but
trajectory-wise results.

With several external players, endogenous support selection adds a
discrete strategic layer to the fixed-influence OSAOC game studied in
\cite{GentilBhaya2026CompetitiveFJ}. Each support-and-action best response
retains the exact sorting solution, and the full binary--continuous stage
game remains an exact potential game. These are frozen-state results: the
present paper does not claim a general convergence or stability theorem
for the competitive closed loop. Same-state welfare comparisons separate
action-level loss from the additional loss caused by strategic support
selection; neither should be confused with residual error imposed by the
sparse feasible family or with cumulative costs generated along different
closed-loop trajectories.

The numerical examples illustrate the same separation.  Wang et al.'s
order-sensitive stochastic matrices provide convergent and periodic
DeGroot benchmarks, while the FJ experiment contrasts invariant and
prejudice-mismatched targets.  Budget and competitive experiments quantify
resource and strategic effects without attributing closed-loop performance
to the one-step sorting theorem itself.

Two limitations are central.  First, the intervention \(bu\) is direct.
As shown in Section~\ref{sec:single_player}, susceptibility-filtered actuation
\(\Lambda bu\) preserves the exact sorting theorem when susceptibility is
homogeneous on the target set, after the rescaling
\(\gamma\mapsto\gamma/\lambda^2\).  With heterogeneous susceptibilities,
the reduced denominator becomes support dependent and the problem is instead
a weighted fractional binary program.  Second, both the block small-gain
certificate and the periodic cycle bounds are sufficient and can be
conservative; the latter are additionally conditional on the realized
closed-loop trajectory.  Extensions to weighted/susceptibility-filtered or
multi-input sparse actuation, sharper graph-structural tracking conditions,
and scalable equilibrium computation remain natural directions.

\section*{Supplementary materials}

Supplementary materials are provided in the \textbf{Supplementary Files} associated with this article. These materials include \path{Supplementary_Network_Matrices.pdf}, which provides the complete network matrices employed in the analyses; and \path{Detailed_Derivations.pdf}, which presents additional mathematical derivations and supporting analytical developments. All these documents are available in the \textbf{Supplementary Files} accompanying this article.

\section*{Acknowledgments}
\noindent This research was partially supported by Coordenação de Aperfeiçoamento de Pessoal de Nível Superior - Brasil (CAPES) Finance Code 001, CNPq BPP/PQ-Sr grant no. 313335/2022-2 [AB] and CNPq doctoral fellowship grant no. 163968/2021-7 [GG].

\bibliography{strategic_affine_ejc_refs}
\bibliographystyle{elsarticle-num}
\end{document}

%% file: strategic_switching_outputs/RESULTTODO_values.tex

\newcommand{\NumSortingTests}{2004}
\newcommand{\MaxSortingError}{7.105\times 10^{-15}}
\newcommand{\SettleTol}{10^{-6}}
\newcommand{\CycleZeroTol}{10^{-9}}
\newcommand{\ActivationValueTol}{10^{-24}}
\newcommand{\BudgetGamma}{0.10}
\newcommand{\BudgetHorizon}{250}
\newcommand{\GammaGrid}{0.05,0.10,0.25,0.50,1,2,5}
\newcommand{\NumRandomSeeds}{100}
\newcommand{\BudgetMCSamples}{50}
\newcommand{\CompetitiveMCSamples}{10}
\newcommand{\OpenConsensusLimit}{0.526667}
\newcommand{\OpenOscMin}{0.305000}
\newcommand{\OpenOscMax}{0.710000}
\newcommand{\LastActiveConsensus}{36}
\newcommand{\LastActiveOscillatory}{51}
\newcommand{\AObsConsensus}{0.041901}
\newcommand{\AObsOscillatory}{0.168755}
\newcommand{\FJBeta}{0.840387}
\newcommand{\FJUniformCycleBound}{0.328079}
\newcommand{\FJEpsilonInvariant}{0.000000}
\newcommand{\FJEpsilonMismatch}{0.167070}
\newcommand{\FJAObsInvariant}{0.003382}
\newcommand{\FJAObsMismatch}{0.020028}
\newcommand{\FJUltimateInvariant}{0.000000}
\newcommand{\FJUltimateMismatch}{0.904809}
\newcommand{\FJFinalCycleInvariant}{0.000000}
\newcommand{\FJFinalCycleMismatch}{0.172672}
\newcommand{\FJMaxTargetInvariant}{0.000000}
\newcommand{\FJMaxTargetMismatch}{0.003023}
\newcommand{\FJLastActiveInvariant}{22}
\newcommand{\FJLastActiveMismatch}{throughout}
\newcommand{\FJBoundConservatism}{5.24}
\newcommand{\BudgetSettleOne}{64}
\newcommand{\BudgetEnergyOne}{5.692854}
\newcommand{\BudgetErrorOne}{14.174938}
\newcommand{\BudgetCardOne}{0.940}
\newcommand{\BudgetCardPreOne}{1.000}
\newcommand{\BudgetSettleTwo}{35}
\newcommand{\BudgetEnergyTwo}{3.215374}
\newcommand{\BudgetErrorTwo}{7.283340}
\newcommand{\BudgetCardTwo}{1.088}
\newcommand{\BudgetCardPreTwo}{2.000}
\newcommand{\BudgetSettleThree}{22}
\newcommand{\BudgetEnergyThree}{2.247954}
\newcommand{\BudgetErrorThree}{5.947485}
\newcommand{\BudgetCardThree}{1.088}
\newcommand{\BudgetCardPreThree}{2.955}
\newcommand{\BudgetSettleFour}{20}
\newcommand{\BudgetEnergyFour}{1.623294}
\newcommand{\BudgetErrorFour}{5.295256}
\newcommand{\BudgetCardFour}{1.132}
\newcommand{\BudgetCardPreFour}{3.300}
\newcommand{\AdaptiveBaselineError}{7.283340}
\newcommand{\AdaptiveBaselineEnergy}{3.215374}
\newcommand{\AdaptiveBaselineSettle}{35}
\newcommand{\AdaptiveExactBaselineError}{7.283340}
\newcommand{\AdaptiveExactBaselineEnergy}{3.215374}
\newcommand{\AdaptiveExactBaselineSettle}{35}
\newcommand{\StaticBaselineError}{9.277319}
\newcommand{\StaticBaselineEnergy}{2.681452}
\newcommand{\StaticBaselineSettle}{36}
\newcommand{\RandomMeanBaselineError}{10.364563}
\newcommand{\RandomMeanBaselineEnergy}{2.468286}
\newcommand{\RandomMeanBaselineSettle}{36.21}
\newcommand{\RandomStdBaselineError}{1.694770}
\newcommand{\RandomStdBaselineEnergy}{0.344927}
\newcommand{\RandomStdBaselineSettle}{1.64}
\newcommand{\UnconstrainedBaselineError}{5.295256}
\newcommand{\UnconstrainedBaselineEnergy}{1.623294}
\newcommand{\UnconstrainedBaselineSettle}{20}
\newcommand{\TargetSetOne}{\{1,2,6,7,9\}}
\newcommand{\TargetSetTwo}{\{3,4,5,7,10\}}
\newcommand{\ZeroCumSocial}{146.228185}
\newcommand{\ZeroCumEnergy}{34.442863}
\newcommand{\ZeroMeanG}{0.0250}
\newcommand{\SequentialCumSocial}{174.068769}
\newcommand{\SequentialCumEnergy}{87.488334}
\newcommand{\SequentialMeanG}{0.9875}
\newcommand{\IBRECumSocial}{174.406471}
\newcommand{\IBRECumEnergy}{88.978959}
\newcommand{\IBREMeanG}{1.0000}
\newcommand{\CentralCumSocial}{107.040499}
\newcommand{\CentralCumEnergy}{26.927910}
\newcommand{\CentralMeanG}{0.0000}
\newcommand{\MaxNEResidual}{2.220\times 10^{-16}}
\newcommand{\CumDeltaAct}{27.528047}
\newcommand{\CumDeltaSup}{41.269954}
\newcommand{\CumDeltaSp}{68.798001}
\newcommand{\MCIBREWorseCount}{10}

%% file: strategic_switching_outputs/table_sorting_validation.tex
\begin{table}[t]
\caption{Validation and timing of exact sparse support selection.}
\label{tab:sorting_validation}
\centering
\resizebox{\columnwidth}{!}{%
\begin{tabular}{r r r r r}
\hline
\(N\) & \(r\) & full sort (ms) & enum. (ms) & max. \(\varepsilon_V\)\\
\hline
8 & 5 & 0.012209 & 0.624367 & 1.78e-15\\
10 & 5 & 0.013831 & 1.795895 & 3.55e-15\\
12 & 5 & 0.017837 & 4.574785 & 2.66e-15\\
14 & 5 & 0.023130 & 10.002675 & 3.55e-15\\
16 & 5 & 0.027576 & 19.544519 & 3.55e-15\\
18 & 5 & 0.034877 & 36.164619 & 7.11e-15\\
\hline
\end{tabular}
}%
\end{table}

%% file: strategic_switching_outputs/five_node_matrices_fragment.tex
\begin{figure*}[t]
\centering
\begingroup
\setlength{\arraycolsep}{3pt}
\small
\[
\begin{array}{c@{\qquad}c@{\qquad}c}
W_1=\begin{bmatrix}
0.5 & 0.5 & 0.0 & 0.0 & 0.0\\
0.5 & 0.5 & 0.0 & 0.0 & 0.0\\
0.5 & 0.5 & 0.0 & 0.0 & 0.0\\
0.0 & 0.0 & 0.0 & 1.0 & 0.0\\
0.0 & 0.0 & 0.0 & 0.0 & 1.0\\
\end{bmatrix} & W_2=\begin{bmatrix}
1.0 & 0.0 & 0.0 & 0.0 & 0.0\\
0.0 & 1.0 & 0.0 & 0.0 & 0.0\\
0.0 & 0.0 & 0.5 & 0.5 & 0.0\\
0.0 & 0.0 & 0.5 & 0.5 & 0.0\\
0.0 & 0.0 & 0.5 & 0.5 & 0.0\\
\end{bmatrix} & W_3=\begin{bmatrix}
0.5 & 0.0 & 0.0 & 0.0 & 0.5\\
0.5 & 0.0 & 0.0 & 0.0 & 0.5\\
0.0 & 0.0 & 1.0 & 0.0 & 0.0\\
0.0 & 0.0 & 0.0 & 1.0 & 0.0\\
0.5 & 0.0 & 0.0 & 0.0 & 0.5\\
\end{bmatrix}\\[2mm]
W_4=\begin{bmatrix}
1.0 & 0.0 & 0.0 & 0.0 & 0.0\\
0.0 & 0.5 & 0.5 & 0.0 & 0.0\\
0.0 & 0.5 & 0.5 & 0.0 & 0.0\\
0.0 & 0.5 & 0.5 & 0.0 & 0.0\\
0.0 & 0.0 & 0.0 & 0.0 & 1.0\\
\end{bmatrix} & W_5=\begin{bmatrix}
0.0 & 0.0 & 0.0 & 0.5 & 0.5\\
0.0 & 1.0 & 0.0 & 0.0 & 0.0\\
0.0 & 0.0 & 1.0 & 0.0 & 0.0\\
0.0 & 0.0 & 0.0 & 0.5 & 0.5\\
0.0 & 0.0 & 0.0 & 0.5 & 0.5\\
\end{bmatrix} &
\end{array}
\]
\endgroup
\caption{Five representative stochastic matrices from the Wang et al. benchmark used in the DeGroot and FJ numerical experiments.}
\label{fig:five_node_matrices}
\end{figure*}

%% file: strategic_switching_outputs/table_fj_specialization.tex
\begin{table}[t]
\caption{Friedkin--Johnsen specialization. The reported \(\chi_m\) is the homogeneous cycle gain, not the full affine cycle-error ratio.}
\label{tab:fj_specialization}
\centering
\begin{tabular}{l r r}
\hline
Metric & invariant & mismatched\\
\hline
\(\varepsilon_g\) & 0.000000 & 0.167070\\
\(\max_m\chi_m\) & 0.003382 & 0.020028\\
cycle UUB & 0.000000 & 0.904809\\
final cycle-start \(\|e\|_2\) & 0.000000 & 0.172672\\
max. final-cycle \(E_{\mathcal T}\) & 0.000000 & 0.003023\\
activity & last \(k=22\) & throughout horizon\\
\hline
\end{tabular}
\end{table}

%% file: strategic_switching_outputs/table_budget_results.tex
\begin{table}[t]
\caption{Effect of the cardinality budget in the deterministic eight-node benchmark (\(\gamma=0.1\), \(K=250\)).}
\label{tab:budget_results}
\centering
\resizebox{\columnwidth}{!}{%
\begin{tabular}{c|cccc}
\hline
\(r\) & 1 & 2 & 3 & 4\\
\hline
\(t_{\rm set}\) & 64 & 35 & 22 & 20\\
\(\mathcal E_u\) & 5.692854 & 3.215374 & 2.247954 & 1.623294\\
\(\sum_k E_{\mathcal T}(k)\) & 14.174938 & 7.283340 & 5.947485 & 5.295256\\
mean \(s(k),\ k<t_{\rm set}\) & 1.000 & 2.000 & 2.955 & 3.300\\
\hline
\end{tabular}
}%
\end{table}

%% file: strategic_switching_outputs/table_budget_monte_carlo.tex
\begin{table}[t]
\caption{Monte-Carlo budget sensitivity over 50 initial conditions, \(x_i(0)\sim\mathcal U[-1.5,1]\).}
\label{tab:budget_mc_results}
\centering
\resizebox{\columnwidth}{!}{%
\begin{tabular}{c r r r}
\hline
\(r\) & median \(t_{\rm set}\) [IQR] & mean \(\sum E_{\mathcal T}\) & mean \(\mathcal E_u\)\\
\hline
1 & 60 [57,62] & 10.730 $\pm$ 5.138 & 3.168 $\pm$ 1.581\\
2 & 29 [27,29] & 6.172 $\pm$ 2.778 & 1.743 $\pm$ 0.865\\
3 & 20 [18,20] & 5.153 $\pm$ 2.293 & 1.226 $\pm$ 0.581\\
4 & 19 [18,19] & 4.950 $\pm$ 2.186 & 1.052 $\pm$ 0.472\\
\hline
\end{tabular}
}%
\end{table}

%% file: strategic_switching_outputs/table_gamma_sweep.tex
\begin{table}[t]
\caption{Regularization sensitivity for \(r=2\) in the deterministic eight-node benchmark.}
\label{tab:gamma_sweep_results}
\centering
\resizebox{\columnwidth}{!}{%
\begin{tabular}{c r r r r}
\hline
\(\gamma\) & \(t_{\rm set}\) & \(\mathcal E_u\) & \(\sum E_{\mathcal T}\) & mean \(s\), pre-settle\\
\hline
0.05 & 33 & 3.311382 & 7.173285 & 2.000\\
0.10 & 35 & 3.215374 & 7.283340 & 2.000\\
0.25 & 37 & 2.958566 & 7.664540 & 2.000\\
0.50 & 43 & 2.612151 & 8.435963 & 2.000\\
1.00 & 52 & 2.118782 & 10.323852 & 2.000\\
2.00 & 71 & 1.540970 & 14.848505 & 2.000\\
5.00 & 104 & 0.740215 & 26.620202 & 2.000\\
\hline
\end{tabular}
}%
\end{table}

%% file: strategic_switching_outputs/table_competitive_summary.tex
\begin{table}[t]
\caption{Competitive one-step benchmarks for the deterministic ten-node network.}
\label{tab:competitive_summary}
\centering
\resizebox{\columnwidth}{!}{%
\begin{tabular}{l r r r}
\hline
Method & cum. social cost & cum. energy & mean \(G_{12}\)\\
\hline
Zero-opponent myopic & 146.228185 & 34.442863 & 0.0250\\
One-sweep sequential BR & 174.068769 & 87.488334 & 0.9875\\
Exact IBRE & 174.406471 & 88.978959 & 1.0000\\
Centralized sparse & 107.040499 & 26.927910 & 0.0000\\
\hline
\end{tabular}
}%
\end{table}

%% file: strategic_switching_outputs/table_welfare_summary.tex
\begin{table}[t]
\caption{Same-state welfare decomposition along the exact-IBRE trajectory.}
\label{tab:welfare_summary}
\centering
\begin{tabular}{l r}
\hline
Quantity & cumulative value\\
\hline
\(\sum_k\Delta_{\rm act}(k)\) & 27.528047\\
\(\sum_k\Delta_{\rm sup}(k)\) & 41.269954\\
\(\sum_k\Delta_{\rm sp}(k)\) & 68.798001\\
max. unilateral IBRE residual & 2.220e-16\\
\hline
\end{tabular}
\end{table}

%% file: strategic_switching_outputs/table_competitive_monte_carlo.tex
\begin{table}[t]
\caption{Competitive sensitivity over ten independently generated periodic networks; entries are mean $\pm$ standard deviation over network realizations.}
\label{tab:competitive_mc}
\centering
\resizebox{\columnwidth}{!}{%
\begin{tabular}{l r}
\hline
Metric & mean $\pm$ sd\\
\hline
zero-opponent cumulative social cost & 139.071 $\pm$ 6.757\\
exact-IBRE cumulative social cost & 164.426 $\pm$ 9.419\\
centralized cumulative social cost & 98.465 $\pm$ 5.972\\
\(\sum\Delta_{\rm act}\) & 28.169 $\pm$ 1.479\\
\(\sum\Delta_{\rm sup}\) & 36.140 $\pm$ 3.577\\
\hline
\end{tabular}
}%
\end{table}

%% file: strategic_switching_outputs/table_baselines.tex
\begin{table}[t]
\caption{Adaptive support selection and baselines (nominal eight-node trajectory, \(r=2\) where applicable).}
\label{tab:baseline_results}
\centering
\resizebox{\columnwidth}{!}{%
\begin{tabular}{l r r r}
\hline
Method & cumulative target error & control energy & \(t_{\rm set}\)\\
\hline
Adaptive \(|S|\le r\) & 7.283340 & 3.215374 & 35\\
Adaptive \(|S|=r\) & 7.283340 & 3.215374 & 35\\
Static initial optimum & 9.277319 & 2.681452 & 36\\
Random exact-\(r\), mean $\pm$ sd & 10.365 $\pm$ 1.695 & 2.468 $\pm$ 0.345 & 36.2 $\pm$ 1.6\\
Unconstrained \(r=|\mathcal T|\) & 5.295256 & 1.623294 & 20\\
\hline
\end{tabular}
}%
\end{table}

%% file: strategic_affine_ejc_refs.bib
@article{DeGroot1974,
  author  = {Morris H. DeGroot},
  title   = {Reaching a Consensus},
  journal = {Journal of the American Statistical Association},
  volume  = {69},
  number  = {345},
  pages   = {118--121},
  year    = {1974},
  doi     = {10.1080/01621459.1974.10480137}
}

@article{FriedkinJohnsen1990,
  author  = {Noah E. Friedkin and Eugene C. Johnsen},
  title   = {Social Influence and Opinions},
  journal = {The Journal of Mathematical Sociology},
  volume  = {15},
  number  = {3--4},
  pages   = {193--206},
  year    = {1990},
  doi     = {10.1080/0022250X.1990.9990069}
}

@article{FriedkinJohnsen1999,
  author  = {Noah E. Friedkin and Eugene C. Johnsen},
  title   = {Social Influence Networks and Opinion Change},
  journal = {Advances in Group Processes},
  volume  = {16},
  pages   = {1--29},
  year    = {1999}
}

@article{ParsegovEtAl2017,
  author  = {Sergey E. Parsegov and Anton V. Proskurnikov and Roberto Tempo and Noah E. Friedkin},
  title   = {Novel Multidimensional Models of Opinion Dynamics in Social Networks},
  journal = {IEEE Transactions on Automatic Control},
  volume  = {62},
  number  = {5},
  pages   = {2270--2285},
  year    = {2017},
  doi     = {10.1109/TAC.2016.2613905}
}

@article{ProskurnikovEtAl2017TimeVarying,
  author  = {Anton V. Proskurnikov and Roberto Tempo and Ming Cao and Noah E. Friedkin},
  title   = {Opinion Evolution in Time-Varying Social Influence Networks with Prejudiced Agents},
  journal = {IFAC-PapersOnLine},
  volume  = {50},
  number  = {1},
  pages   = {11896--11901},
  year    = {2017},
  doi     = {10.1016/j.ifacol.2017.08.1424}
}

@article{ProskurnikovTempo2017,
  author  = {Anton V. Proskurnikov and Roberto Tempo},
  title   = {A Tutorial on Modeling and Analysis of Dynamic Social Networks. Part I},
  journal = {Annual Reviews in Control},
  volume  = {43},
  pages   = {65--79},
  year    = {2017}
}

@article{ProskurnikovTempo2018,
  author  = {Anton V. Proskurnikov and Roberto Tempo},
  title   = {A Tutorial on Modeling and Analysis of Dynamic Social Networks. Part II},
  journal = {Annual Reviews in Control},
  volume  = {45},
  pages   = {166--190},
  year    = {2018}
}

@article{AndersonYe2019,
  author  = {Brian D. O. Anderson and Mengbin Ye},
  title   = {Recent Advances in the Modelling and Analysis of Opinion Dynamics on Influence Networks},
  journal = {International Journal of Automation and Computing},
  volume  = {16},
  pages   = {129--149},
  year    = {2019},
  doi     = {10.1007/s11633-019-1169-8}
}

@article{TianWang2023,
  author  = {Ye Tian and Long Wang},
  title   = {Dynamics of Opinion Formation, Social Power Evolution, and Naive Learning in Social Networks},
  journal = {Annual Reviews in Control},
  volume  = {55},
  pages   = {182--193},
  year    = {2023},
  doi     = {10.1016/j.arcontrol.2023.04.001}
}

@article{DisaroValcher2024,
  author  = {Giorgia Disar{\`o} and Maria Elena Valcher},
  title   = {Balancing Homophily and Prejudices in Opinion Dynamics: An Extended {F}riedkin--{J}ohnsen Model},
  journal = {Automatica},
  volume  = {166},
  pages   = {111711},
  year    = {2024},
  doi     = {10.1016/j.automatica.2024.111711}
}

@article{Masuda2015,
  author  = {Naoki Masuda},
  title   = {Opinion Control in Complex Networks},
  journal = {New Journal of Physics},
  volume  = {17},
  number  = {3},
  pages   = {033031},
  year    = {2015}
}

@article{BiniEtAl2022,
  author  = {Matteo Bini and Paolo Frasca and Chiara Ravazzi and Fabrizio Dabbene},
  title   = {Graph Structure-Based Heuristics for Optimal Targeting in Social Networks},
  journal = {IEEE Transactions on Control of Network Systems},
  volume  = {9},
  number  = {3},
  pages   = {1189--1201},
  year    = {2022}
}

@inproceedings{SunZhang2023,
  author    = {Haoxin Sun and Zhongzhi Zhang},
  title     = {Opinion Optimization in Directed Social Networks},
  booktitle = {Proceedings of the AAAI Conference on Artificial Intelligence},
  volume    = {37},
  number    = {4},
  pages     = {4623--4632},
  year      = {2023},
  doi       = {10.1609/aaai.v37i4.25585}
}

@article{SunZhang2026,
  author  = {Haoxin Sun and Zhongzhi Zhang},
  title   = {Leader Selection for Opinion Optimization in Social Networks},
  journal = {Theoretical Computer Science},
  volume  = {1073},
  pages   = {115894},
  year    = {2026},
  doi     = {10.1016/j.tcs.2026.115894}
}

@article{BastopcuEtAl2025,
  author  = {Melih Bastopcu and S. Rasoul Etesami and Tamer Ba{\c{s}}ar},
  title   = {Online and Offline Dynamic Influence Maximization Games over Social Networks},
  journal = {IEEE Transactions on Control of Network Systems},
  volume  = {12},
  number  = {2},
  pages   = {1440--1453},
  year    = {2025}
}

@book{KaszkurewiczBhaya2022,
  author    = {Eugenius Kaszkurewicz and Amit Bhaya},
  title     = {Business Dynamics Models: Optimization-Based One Step Ahead Optimal Control},
  series    = {Advances in Design and Control},
  publisher = {Society for Industrial and Applied Mathematics},
  address   = {Philadelphia, PA},
  year      = {2022},
  doi       = {10.1137/1.9781611977318},
  isbn      = {978-1-61197-730-1}
}

@article{GentilBhaya2024,
  author  = {Gabriel Gentil and Amit Bhaya},
  title   = {Opinion Dynamic Games under One Step Ahead Optimal Control},
  journal = {IEEE Transactions on Computational Social Systems},
  volume  = {11},
  number  = {3},
  pages   = {4202--4213},
  year    = {2024},
  note    = {Correction in vol. 11, no. 4, p. 5554}
}

@article{WangEtAl2023ConcatenatedFJ,
  author  = {Lingfei Wang and Carmela Bernardo and Yiguang Hong and Francesco Vasca and Guodong Shi and Claudio Altafini},
  title   = {Consensus in Concatenated Opinion Dynamics With Stubborn Agents},
  journal = {IEEE Transactions on Automatic Control},
  volume  = {68},
  number  = {7},
  pages   = {4008--4023},
  month   = jul,
  year    = {2023},
  doi     = {10.1109/TAC.2022.3200888}
}

@misc{GentilBhaya2026CompetitiveFJ,
  author        = {Gabriel Gentil and Amit Bhaya},
  title         = {Competitive One-Step-Ahead Control of {F}riedkin--{J}ohnsen Networks: Potential Games, Stability, and the Price of Competition},
  year          = {2026},
  eprint        = {2608.27623},
  archivePrefix = {arXiv},
  primaryClass  = {eess.SY},
  note          = {arXiv:2608.27623 [eess.SY]},
  url           = {https://arxiv.org/abs/2608.27623}
}

@article{LiuMazalovGao2026EJC,
  author  = {Yanshan Liu and Vladimir V. Mazalov and Hongwei Gao},
  title   = {Comparison of Control Strategies in an Average-Oriented Opinion Dynamics Game},
  journal = {European Journal of Control},
  volume  = {91},
  pages   = {101551},
  month   = sep,
  year    = {2026},
  doi     = {10.1016/j.ejcon.2026.101551}
}
